\documentclass{article}
\usepackage[utf8]{inputenc}
\usepackage{mystyle}

\usepackage{mathtools}

\usepackage[colorlinks = true, allcolors = blue!60!gray]{hyperref}

\newcommand{\J}{N} %for players
\newcommand{\M}{E} %for machines/resources
\newcommand{\vb}{v_0}

\title{Selfish, Local and Online Scheduling via Vector Fitting}
\author{Danish Kashaev \footnote{\texttt{danish.kashaev@cwi.nl}, Centrum Wiskunde \& Informatica, Amsterdam}}

\date{}

\begin{document}
\maketitle

\begin{abstract}
We provide a dual fitting technique on a semidefinite program yielding simple proofs of tight bounds for the robust price of anarchy of several congestion and scheduling games under the sum of weighted completion times objective. The same approach also allows to bound the approximation ratio of local search algorithms and the competitive ratio of online algorithms for the scheduling problem $R || \sum w_j C_j$. All of our results are obtained through a simple unified dual fitting argument on the same semidefinite programming relaxation, which can essentially be obtained through the first round of the Lasserre/Sum of Squares hierarchy.

As our main application, we show that the known coordination ratio bounds of respectively $4, (3 + \sqrt{5})/2 \approx 2.618,$ and  $32/15 \approx 2.133$ for the scheduling game $R || \sum w_j C_j$ under the coordination mechanisms Smith's Rule, Proportional Sharing and Rand (STOC 2011) can be extended to congestion games and obtained through this approach. For the natural restriction where the weight of each player is proportional to its processing time on every resource, we show that the last bound can be improved from 2.133 to 2. This improvement can also be made for general instances when considering the price of anarchy of the game, rather than the coordination ratio. As a further application of this technique in a game theoretic setting, we show that it recovers the tight bound of $(3 + \sqrt{5})/2$ for the price of anarchy of weighted affine congestion games and the Kawaguchi-Kyan bound of $(1+ \sqrt{2})/2$ for the pure price of anarchy of $P || \sum w_j C_j$. 

Moreover, we show that this approach recovers the known tight approximation ratio of $(3 + \sqrt{5})/2$ for a natural local search algorithm for $R || \sum w_j C_j$, as well as the best currently known combinatorial approximation algorithm for this problem achieving an approximation ratio of $(5 + \sqrt{5})/4 + \varepsilon \approx 1.809 + \varepsilon$, and provide an almost matching lower bound.

Finally, we show that this technique also extends to online algorithms by analyzing the greedy algorithm for $R || \sum w_j C_j$ achieving a competitive ratio of $4$ in an online setting where the arrival order of the jobs is adversarial and the ordering on each machine is optimal.
\end{abstract}

\section{Introduction}
A standard way of quantifying inefficiency of selfish behaviour in algorithmic game theory is the price of anarchy, introduced in \cite{koutsoupias1999worst}. It is defined as the ratio between the cost of a worst-case Nash equilibrium and the cost of a social optimum. This definition can be used to understand inefficiency of pure or mixed Nash equilibria, and can also be extended to more general notions, such as correlated or coarse-correlated equilibria. 

Developing tools to bound the price of anarchy is a central question, and several approaches have been proposed in the literature to tackle this problem. One technique that has been very successful for a variety of games is the smoothness framework, introduced in \cite{roughgarden2015intrinsic}. One advantage of this approach is that it automatically bounds the price of anarchy for all the different notions of equilibria mentioned above, yielding bounds on the robust price of anarchy of a game \cite{roughgarden2015intrinsic}.

Another possible avenue is to use convex relaxations to help bound the price of anarchy, as done in \cite{kulkarni2014robust}. The high-level approach is to formulate a convex relaxation of the underlying optimization problem of a given game, and to construct a feasible solution to the dual of that relaxation, whose cost can then be compared to the cost of an equilibrium. Bounding the ratio between the cost of the equilibrium and of the feasible dual solution then yields an upper bound on the price of anarchy by weak duality. 

In this paper, we build on this approach and show that a single convex semidefinite programming relaxation can be used to obtain tight (robust) price of anarchy bounds for several different congestion and scheduling games. This relaxation can in fact be obtained using the first round of the Lasserre hierarchy \cite{lasserre2001global}, and the proofs bounding the price of anarchy through the dual of that relaxation are surprisingly simple and essentially follow the same template for all the games considered. In addition to bounding the price of anarchy, it turns out that the same approach also allows to bound the approximation ratio of local search algorithms and the competitive ratio of online algorithms for machine scheduling.

As a main illustration of this technique, we consider the following model of congestion games. We are given a set of players $N$ and a set of resources $E$. The strategy set for each player $j \in N$ is a collection of subsets of resources and is denoted by $\mathcal{S}_j \subseteq 2^E$. Each player has a resource-dependent processing time $p_{ej} \geq 0$ and a weight $w_j \geq 0$. Once each player chooses a strategy, if a given resource $e \in E$ is shared by several players, then $e$ uses a \emph{coordination mechanism}, defined as a local policy for each resource, in order to process the players using it. One natural example of such a coordination mechanism is to order the players by increasing Smith ratios, defined as the ratio between the processing time on a resource and the weight of a given player \cite{smith1956various}. 

This model is a generalization of the unrelated machine scheduling game $R || \sum w_j C_j$, where each job needs to selfishly pick a machine to minimize its own weighted completion time, while knowing that each machine uses a coordination mechanism to process the jobs assigned to it. In our model, the set of resources $E$ is the set of machines, and the strategy set of each player is a subset of the machines. An important special case of our model, which generalizes $R || \sum w_j C_j$, is the following selfish routing game. We are given a directed graph $G = (V,E)$ and a set of players $N$. Each player $j$ wants to pick a path between a source node $s_j \in V$ and a sink node $t_j \in V$. The strategy set $\mathcal{S}_j$ for player $j \in N$ is the set of all paths between $s_j$ and $t_j$. A parallel link network where each player has the same source and sink node exactly corresponds to the $R || \sum w_j C_j$ scheduling problem. 

The work of \cite{cole2011inner} considers three different coordination mechanisms for $R || \sum w_j C_j$. Their main results are that \emph{Smith's Rule} leads to a tight price of anarchy of $4$, and this can be improved to $(3 + \sqrt{5})/2 \approx 2.618$ and $32/15 \approx 2.133$ by respectively considering a preemptive mechanism called \emph{Proportional Sharing}, as well as a randomized one named \emph{Rand}. The latter two results in fact bound the \emph{coordination ratio} of the coordination mechanism, meaning that the cost of a worst-case Nash equilibrium is compared to the cost of an optimal solution under \emph{Smith's Rule}, since this is always how an optimal solution processes the jobs once an assignment is given \cite{smith1956various}. The proof technique they use to obtain their results is based on the smoothness framework \cite{roughgarden2015intrinsic}. In order to exploit the structure of the problem, they map strategy vectors into a carefully chosen inner product space, where the social cost is closely related to a squared norm in that space. Generalizing their results to selfish routing games was mentioned as an open question.

The inner product space structure developed in \cite{cole2011inner} turns out to have a natural connection to semidefinite programming, since the latter can be seen as optimizing over inner products of vectors. In this work, we study this connection and show that it leads to simple dual fitting proofs that allow to tightly bound the price of anarchy for several different congestion and scheduling games in a unified way. Moreover, it also allows to bound the approximation ratio of local search algorithms, as well as the competitive ratio of online algorithms for such problems. We hope that this new approach might turn out to be useful in other contexts as well.

\subsection*{Our contributions}
Our main contribution is a unified dual fitting technique on a single semidefinite program (SDP) to bound the price of anarchy of games, the approximation ratio of local search algorithms, and the competitive ratio of online algorithms for problems whose underlying optimization problem can be cast as a binary quadratic program. We illustrate the applicability of this approach for different scheduling and congestion problems. The semidefinite program used can be obtained by applying one round of the Lasserre/Sum of Squares hierarchy to the exact binary quadratic program. 

We show that the three coordination ratio bounds of respectively $4, (3 + \sqrt{5})/2 \approx 2.618$ and $32/15 \approx 2.133$ for the policies \emph{Smith's Rule}, \emph{Proportional Sharing} and \emph{Rand} can be obtained through our approach in the above congestion game model. This yields alternative and simple proofs of these results in a more general model, which avoid the use of minimum norm distortion inequalities, as done in \cite{cole2011inner}. We moreover show that the last bound can be improved from 2.133 to 2 for the natural special case where the processing times are proportional to the weight of a given player on every feasible resource. This means that every resource has a real-value $\lambda_e \geq 0$, interpreted as the processing power, and the processing time of every player satisfies $p_{ej} \in \{\lambda_e w_j, \infty\}$ for every $e \in E, j \in N$. The importance of this model in a scheduling setting has been mentioned in \cite{kalaitzis2017unrelated}. This improvement from 2.133 to 2 can also be obtained for general instances if one considers the price of anarchy of the game, rather than the coordination ratio. This means that the cost of a worst-case Nash equilibrium is now compared against an optimal solution using the \emph{Rand} policy, rather than \emph{Smith's Rule}.

Moreover, we show that the same approach (on the same relaxation) can be used to bound the approximation ratio of local search algorithms for machine scheduling under the sum of weighted completion times objective. We first consider a natural algorithm whose local optima simply ensure that no job can decrease the global objective function by switching to a different machine. Observe the analogy with Nash equilibria, which ensure that no job can improve its own objective (or completion time) by switching machines. We recover the approximation ratios of $(3 + \sqrt{5})/2 \approx 2.618$ and $(5 + \sqrt{5})/4 \approx 1.809$ for the scheduling problems $R || \sum w_j C_j$ and $P | \mathcal{M}_j | \sum w_j C_j$ given in \cite{correa2022performance}. In addition, we also analyze an improved local search algorithm for $R || \sum w_j C_j$ attaining a bound of $(5 + \sqrt{5})/4 + \varepsilon \approx 1.809 + \varepsilon$ \cite{caragiannis2017coordination}, and show an almost matching lower bound of $1.791$. To the best of our knowledge, this is the currently best known \emph{combinatorial} approximation algorithm for this problem.

As a further illustration of the technique in a game theoretic setting, we apply it to two classical games and show that it yields simple proofs of known tight price of anarchy bounds. We first show how to get the tight bound of $(3+ \sqrt{5})/2$ for the price of anarchy of weighted affine congestion games. While a dual fitting proof through a convex relaxation of this bound is already provided in \cite{kulkarni2014robust}, this result showcases the versatility of our SDP relaxation and of the fitting strategy. In addition, a dual fitting proof of the Kawaguchi-Kyan bound of $(1 + \sqrt{2})/2$ for the pure price of anarchy of the scheduling game $P || \sum w_j C_j$ is also provided through the same relaxation. We note that the dual fitting strategy used for this result uses a reduction to worst-case instances of \cite{schwiegelshohn2011alternative}.

Finally, we also study the $R || \sum w_j C_j$ scheduling problem in the following online setting. A set of jobs arrives online in an adversarial order. Whenever a job arrives, an online algorithm needs to irrevocably assign it to a machine, at which point the job enters the schedule on that machine in the correct position with respect to the Smith ratio. This can equivalently be seen as the arrival order of the jobs as being online, but the ordering on each machine as being an offline decision. It is known that the greedy algorithm achieves a competitive ratio of $4$ \cite{gupta2020greed}. We show how to analyze this algorithm in our congestion model using our SDP dual fitting approach, illustrating how the technique can be adapted to online algorithms.

\subsection*{Further related work}
There is a vast literature on exact or approximation algorithms for scheduling problems under the (weighted) sum of completion times objective. We adopt the standard three-field notation $\alpha|\beta|\gamma$ of \cite{graham1979optimization}. The problem with unweighted completion times $R||\sum C_j$ is polynomial time solvable \cite{horn1973minimizing, bruno1974scheduling}. For $P||\sum C_j$ on parallel machines, the shortest first policy gives an optimal solution which also turns out to be a Nash equilibrium \cite{conway1967miller}. On the other hand, the weighted completion times objective is NP-hard even for $P || \sum w_j C_j$ \cite{lenstra1977complexity}. A PTAS is known for $P || \sum w_j C_j$ \cite{skutella2000ptas}, while $R || \sum w_j C_j$ is APX-hard \cite{hoogeveen1998non}. Constant factor approximation algorithms are however possible, with major results being a simple $3/2$-approximation by rounding a convex relaxation \cite{skutella2001convex, sethuraman1999optimal} and the first algorithm breaking the $3/2$-approximation using a semidefinite relaxation \cite{bansal2016lift}. We note that the primal semidefinite program used in our paper is very similar to their relaxation. Building on this, subsequent improvements have been made \cite{im2020weighted, im2023improved, harris2024dependent} with the current best (to the best of our knowledge) approximation algorithm for this problem obtaining a ratio of $1.36 + \varepsilon$ \cite{li2024approximating}. In the special case where Smith ratios are uniform, an improved bound of $(1+\sqrt{2})/2 + \varepsilon$ has been obtained \cite{kalaitzis2017unrelated}.

Scheduling problems have also been vastly studied from a game theoretic perspective. For $P || \sum w_j C_j$, the pure price of anarchy of \emph{Smith's Rule} coincides with the approximation ratio of a simple greedy algorithm and was shown to be $(1 + \sqrt{2})/2 \approx 1.207$ in a classic result of \cite{kawaguchi1986worst}. A much simpler proof of this result is shown in \cite{schwiegelshohn2011alternative}. Interestingly, the mixed price of anarchy of this game is higher, with a tight bound of $3/2$ even for $P || \sum C_j$ \cite{rahn2013bounding}. For the unweighted version, \emph{Smith's Rule} in fact reduces to the shortest processing time first policy, under which \cite{hoeksma2011price} shows an upper and lower bound of respectively $2$ for the robust price of anarchy and $e/(e-1) \approx 1.58$ for the pure price of anarchy of $Q || \sum C_j$. For related machines, it is still an interesting open question whether the upper bounds of respectively $2$ and $4$ for $Q || \sum C_j$ and $Q || \sum w_j C_j$ can be improved.

Coordination mechanisms were introduced in the work of \cite{christodoulou2004coordination} for $P || C_{max}$ and a selfish routing/congestion game. Four different scheduling games under four different policies were analyzed in \cite{immorlica2009coordination} under the makespan objective. Upper and lower bounds for different coordination mechanisms for $R || C_{max}$ can be found in \cite{azar2008almost, caragiannis2013efficient, fleischer2010preference, cohen2011non, abed2012preemptive}. Further work on coordination mechanisms for the makespan objective has been done in \cite{bhattacharya2014coordination, caragiannis2019almost, kollias2013nonpreemptive}.

The literature for the sum of completion times objective is somewhat sparser. The work of \cite{cole2011inner} considers $R || \sum w_j C_j$ and shows that the policies \emph{Smith's Rule}, \emph{Proportional Sharing} and \emph{Rand} respectively give bounds of $4, 2.618$ and $2.133$ on the robust price of anarchy. The first two bounds are tight, with matching lower bounds given in \cite{correa2012efficiency} and \cite{caragiannis2006tight}. The latter two coordination mechanisms can in fact be interpreted as a cost-sharing protocol \cite{caragiannis2017coordination}. Using similar techniques, \cite{abed2014optimal} extend some of the previous results to multi-job scheduling games. Coordination mechanisms for a more general model with release dates and assignment costs have been studied in \cite{bhattacharya2014coordination}.

The study of the price of anarchy for weighted congestion games was initiated in \cite{koutsoupias1999worst} for parallel links under the maximum load (or makespan) objective. Tight bounds for parallel links have been shown in \cite{czumaj2007tight}. For general networks under the MinSum objective with affine latency functions, the works of \cite{awerbuch2005price, christodoulou2005price} establish that the price of anarchy is $5/2$ for the unweighted version and $(3 + \sqrt{5})/2$ in the weighted case. Other models have been studied in \cite{aland2011exact,caragiannis2006tight, bhawalkar2014weighted, suri2004selfish, farzad2008priority}. To the best of our knowledge, the literature on coordination mechanisms for congestion/selfish routing games is relatively sparse \cite{christodoulou2004coordination, christodoulou2014improving, bhattacharya2014coordinationtree}.

\section{Preliminaries}

\paragraph{Game format.} All the games/problems considered in this paper are of the following form. A set of players $N$ is given. Each player $j \in N$ has a strategy set $\mathcal{S}_j$, and we denote by $x_{ij} \in \{0,1\}$ the binary value indicating whether the player chooses strategy $i \in \mathcal{S}_j$. If $x_{ij} \in [0,1]$, then this corresponds to the probability of player $j$ independently choosing strategy $i$. The (expected) cost incurred by the player is denoted by $C_j(x)$ and is a \emph{quadratic} (possibly non-convex) function of $x$. Given weights $w_j \geq 0$ for every player $j \in N$, the total social cost is the weighted sum of costs incurred by every player, and we denote it by $C(x) = \sum_{j} w_j C_j(x)$. 

\paragraph{Scheduling.} One example falling in this class are scheduling games. Given is a set of jobs $J = N$, which are the players, and a set of machines $M$. The strategy set of every player is a subset of the machines $\mathcal{S}_j \subseteq M$. We adopt the standard three-field notation $\alpha|\beta|\gamma$ of \cite{graham1979optimization}, with $\alpha$ denoting the machine environment, $\beta$ denoting the constraints, and $\gamma$ denoting the objective function. The most general such problem we consider is $R || \sum w_j C_j$, where each job $j \in N$ has unrelated processing times $p_{ij} \in \mathbb{R}_+ \cup \{\infty\}$ for each $i \in M$. If $p_{ij} = \infty$, we will without loss of generality assume that $i \notin \mathcal{S}_j$. Once an assignment $x$ is fixed, the optimal way to process the jobs for each machine is to order them by increasing Smith ratios, which we denote as $\delta_{ij} := p_{ij}/w_j$. We denote $k \prec_i j$ if $k$ precedes $j$ in the ordering of machine $i$, meaning that $\delta_{ik} \leq \delta_{ij}$. We assume ties are broken in a consistent way. The completion time of every job is then 
\[C_j(x) = \sum_{i \in M} x_{ij} \Big(p_{ij} + \sum_{k \prec_i j} p_{ik} x_{ik}\Big).\]
Observe that this is indeed a quadratic function in $x$. If every job has the same processing time $p_{ij} = p_j$ on every machine, this model is denoted by $P || \sum w_j C_j$. If $p_{ij} \in \{p_j, \infty\}$, then the model is denoted as $P | \mathcal{M}_j | \sum w_j C_j$.

\paragraph{Congestion model.} We consider the following model of congestion games, which generalize the scheduling games described above. We are given a set of players $N$ and a set of resources $E$. The strategy set for each player $j \in N$ is denoted by $\mathcal{S}_j \subseteq 2^E$ and is a collection of subsets of resources. Each player has a resource-dependent processing time $p_{ej} \geq 0$ and a weight $w_j \geq 0$. Without loss of generality, we assume that for every feasible strategy $i \in \mathcal{S}_j$ of a player $j \in N$, we have that $p_{ej} < \infty$ for every $e\in i$ (otherwise simply remove $i$ from $\mathcal{S}_j$ since it is not a valid strategy). The Smith ratio is defined as $\delta_{ej} = p_{ej}/w_j$ for every $e \in E, j \in N$. We denote $k \prec_e j$ if $\delta_{ek} \leq \delta_{ej}$, meaning that $k$ has a smaller Smith ratio than $j$ on the resource $e \in E$, where ties are broken in a consistent manner. For a given assignment $(x_{ij})_{j \in N, i \in \mathcal{S}_j}$, we denote
\[z_{ej} := \sum_{i \in \mathcal{S}_j : e \in i} x_{ij}.\]
We invite the reader to think about pure assignments. In that case, $x_{ij} \in \{0,1\}$ is binary and indicates whether or not player $j$ chooses strategy $i \in \mathcal{S}_j$, whereas $z_{ej} \in \{0,1\}$ takes value one if $j$ uses the resource $e \in E$, i.e. chooses a strategy $i \in \mathcal{S}_j$ containing resource $e \in E$. In the case of mixed assignments, $x_{ij} \in [0,1]$ represents the probability of player $j$ independently choosing strategy $i$, whereas $z_{ej} \in [0,1]$ represents the probability of player $j$ using resource $e$. Once an assignment $x$ is fixed, Smith's Rule is again the optimal way for every resource to process the jobs, and the cost incurred by a player $j \in N$ is given by:
\[C_j(x) = \sum_{i \in S_j} x_{ij} \sum_{e \in i}  \Big( p_{ej} +  \sum_{k \prec_e j} p_{ek} \: z_{ek} \Big). \]

\paragraph{Nash equilibria.} An assignment $x$ is a Nash equilibrium if no player can get a lower cost by changing his/her strategy. The price of anarchy is defined as the ratio between the cost of a worst-case Nash equilibrium and the cost of an optimal solution. Unless explicitly stated otherwise, we consider mixed Nash equilibria, meaning that the following set of constraints is satisfied:
\begin{equation}
\label{eq_nash_conditions_gen}
C_j(x) \leq C_j(x_{-j}, i) \qquad \forall j \in N, \forall i \in \mathcal{S}_j
\end{equation}
where $x_{-j}$ refers to the assignment of all players other than $j$. In Appendix \ref{section_robust}, we show how to extend our results to a more general equilibrium concept, namely a \emph{coarse-correlated} equilibrium.

\paragraph{Coordination mechanisms.} In the scheduling setting, a coordination mechanism is a set of local policies, one for each machine, deciding on how the jobs assigned to it should be processed. Smith's Rule is one example of such a policy, which is in fact optimal in terms of the social cost. However, picking a different policy may help improve the price of anarchy. One policy considered in this paper is a preemptive mechanism called \emph{Proportional Sharing}, where the jobs are scheduled in parallel, with each uncompleted job receiving a fraction of the processor time proportional to its weight. Another mechanism is \emph{Rand}, which orders the jobs randomly by ensuring that the probability of job $j$ being scheduled before $k$ is $\delta_{ik}/(\delta_{ij} + \delta_{ik})$ for every pair of jobs $j,k$. The reader is referred to \cite{cole2011inner} for details. In our congestion model, each resource uses one of these coordination mechanisms to process the players using that resource. Note that this modifies the cost $C_j(x)$ incurred by every player, and thus also the social cost $C(x)$. 

\paragraph{Coordination ratio and price of anarchy.} We make a distinction between the coordination ratio of a coordination mechanism and the price of anarchy of the game for our congestion model. The coordination ratio measures the ratio between a worst-case Nash equilibrium and an optimal solution if every resource uses Smith's Rule to process the players. In contrast, the price of anarchy of the game compares to a weaker optimal solution where each resource uses the chosen mechanism to process the players.

\paragraph{Outline of the paper.} The semidefinite programming relaxation and a high-level view of the approach is presented in Section \ref{section_SDP}. The analysis of the coordination ratio and the price of anarchy of \emph{Smith's Rule}, \emph{Proportional Sharing} and \emph{Rand} for our congestion model are presented in Section \ref{sec_cong_coor}. The analysis for the price of anarchy of weighted affine congestion games is shown in Section \ref{sec_weig_cong}. The analysis of local optima for machine scheduling is done in Section \ref{sec_local_opt}. The analysis of the greedy online algorithm is shown in Section \ref{sec_greedy_online}. The pure price of anarchy of $P || \sum w_j C_j$ is presented in Appendix \ref{sec_kawa_kyan}.

\section{The semidefinite programming relaxation}
\label{section_SDP}
\subsection{The primal-dual pair}
\label{sec_SDP}
We assume in this section some basics on semidefinite programs (SDPs), which can be found in Appendix \ref{sec_sdp_basics}. Let $N$ be a set of players, with each player $j \in N$ having a strategy set $\mathcal{S}_j$. An exact quadratic program to compute the social optimum is given by 

\begin{align*}
    \min \:  &C(x) \\
    &\sum_{i \in \mathcal{S}_j} x_{ij} = 1 \qquad \forall j \in N \\
    &x_{ij} \in \{0,1\} \qquad \forall j \in N, \forall i \in \mathcal{S}_j.
\end{align*}

\noindent Since we assume $C(x)$ to be quadratic in $x$, the social cost can be written as 
\begin{equation}
\label{eq_soc_cost_quadratic}
C(x) = C_{\{0,0\}} + 2\sum_{j \in N, i \in \mathcal{S}_j} C_{\{0, ij\}}\:x_{ij} + \sum_{\substack{j, k \in N\\ i \in \mathcal{S}_j, i' \in \mathcal{S}_k}} C_{\{ij, \: i'k\}}x_{ij} \: x_{i'k}
\end{equation}
for some symmetric matrix $C$ of dimension $1 + \sum_{j \in N}|\mathcal{S}_j|$, which has one row/column corresponding to each $x_{ij}$, in addition to one extra row/column that we index by $0$. The above equation \eqref{eq_soc_cost_quadratic} can be written in a compact way as $C(x) = \langle C, X \rangle := \text{Tr} (C^T X)$, where $X$ is the rank one matrix $X = (1,x) (1,x)^T$, where the notation $(1,x)$ refers to a vector in dimension $1 + \sum_{j \in N}|\mathcal{S}_j|$ obtained by appending a coordinate with value $1$ to $x$. 

We now consider a semidefinite convex relaxation of the above quadratic program, which can essentially be obtained through the Lasserre/Sum of Squares hierarchy \cite{lasserre2001global}. The variable of the program is a positive semidefinite matrix $X$ of dimension $1 + \sum_{j \in N}|\mathcal{S}_j|$, which has one row/column corresponding to each $x_{ij}$, in addition to one extra row/column that we index by $0$. 

\begin{align*}
    \min \langle C, X \rangle \qquad  \qquad &\\
    \sum_{i \in \mathcal{S}_j} X_{\{ij,\: ij\}} &= 1 \hspace{4cm} \forall j \in N \\
    X_{\{0,0\}} & = 1 \\
    X_{\{0, \: ij\}} &= X_{\{ij, \: ij\}} \hspace{3cm} \forall j \in N, i \in \mathcal{S}_j \\
    X_{\{ij, \: i'k\}} &\geq 0 \hspace{4cm} \forall {(i,j), (i',k)} \text{ with } j,k \in N. \\
    X & \succeq 0
\end{align*}
To see that this is in fact a relaxation to the previous quadratic program, note that for any binary feasible assignment $x$, the rank-one matrix $X = (1,x) (1, x)^T$ is a feasible solution to the SDP with the same objective value. The key observation that makes this work is the fact that $x_{ij}^2 = x_{ij}$ for $x_{ij} \in \{0,1\}$, leading to $X_{\{ij, \: ij\}} = x_{ij}^2 = x_{ij} = X_{\{0, \: ij\}}$. The dual to this relaxation, written in vector form, is the following. The computation of the dual is shown in Appendix \ref{sec_take_dual}. We call this relaxation \emph{(SDP-C)}.

\vspace{0.1cm}
\begin{align}
\label{dual_sdp}
    \max \sum_{j \in N} y_j  - &\frac{1}{2}\Vert v_0 \Vert ^ 2 \\
    y_j &\leq C_{\{ij, \; ij\}} - \frac{1}{2}\Vert v_{ij} \Vert ^ 2 + \: \langle v_0, v_{ij} \rangle \qquad \qquad \forall j \in N, i \in \mathcal{S}_j \nonumber \\
    \langle v_{ij}, v_{i'k} \rangle &\leq 2 \: C_{\{ij, \: i'k\}} \hspace{4.3cm} \forall (i,j) \neq (i',k) \text{ with } j,k \in N. \nonumber
\end{align}
\vspace{-0.2cm}

\noindent The variables of this program are real-valued $y_j \in \mathbb{R}$ for every $j \in N$, as well as vectors $v_0 \in \mathbb{R}^d$ and $v_{ij} \in \mathbb{R}^d$ for every $j \in N, i \in \mathcal{S}_j$ in some dimension $d \in \mathbb{N}$. This will be the relaxation used for every dual fitting argument in this paper. Depending on the problem we are considering, the matrix $C$, which only depends on the total social cost, is then picked accordingly. The computation of this matrix for every game considered is shown in Appendix \ref{sec_comp_dual}.

\subsection{High-level view of the approach and intuition of the dual}
We describe here a high-level view of the dual fitting approach and of its main ideas. We also provide some intuition in how the dual program \eqref{dual_sdp} is used. For clarity, we illustrate the concepts on a simple concrete toy example: a weighted load balancing game, which is a special case of an affine weighted congestion game later analyzed in full detail in Section \ref{sec_weig_cong}. 

\paragraph{Example: load balancing.} We are given a set of players $N$ and a set of resources $E$. The strategy set of every player $j \in N$ is a subset of resources $\mathcal{S}_j \subseteq E$ with unrelated weights $w_{ij} \geq 0$ associated for every $i \in \mathcal{S}_j$. Consider a pure assignment $x$, the \emph{load} of a resource $i \in E$ is defined as the total weight of players assigned to it and is formally defined as $\ell_i(x) = \sum_{j \in N} w_{ij} x_{ij}$. The cost of a player $j$ is then defined as $C_j(x) = \sum_{i \in E} \ell_i(x) \: w_{ij} \: x_{ij}$, meaning that it is the weight multiplied by the load of the resource picked. The social cost can be written as follows
\begin{equation}
\label{eq_social_cost_load_balancing}
C(x) = \sum_{j \in N} C_j(x) = \sum_{i \in E} \sum_{j, k \in N} w_{ij} \: w_{ik} \: x_{ij} \: x_{ik} = \sum_{\substack{j, k \in N\\ i \in \mathcal{S}_j, i' \in \mathcal{S}_k}} w_{ij} \: w_{ik} \: x_{ij} \: x_{i'k} \: \mathds{1}_{\{i = i'\}}.
\end{equation}
Note that the social cost can also be written in a simple way as $C(x) = \sum_{i \in E} \ell_i(x)^2$. The above equation is however written in the form \eqref{eq_soc_cost_quadratic}.

\paragraph{Specializing the dual SDP.} After understanding what the social cost looks like as a quadratic function in the form \eqref{eq_soc_cost_quadratic}, we are able to write down the dual program \eqref{dual_sdp} for a considered game. In our example, the above equation tells us that the matrix $C$ has diagonal entries $C_{\{ij, \; ij\}} = w_{ij}^2$ and non-diagonal entries $C_{\{ij, \: i'k\}} = w_{ij} \: w_{ik} \: \mathds{1}_{\{i = i'\}}$, meaning that we can write down the dual as:

\begin{align}
    \max \sum_{j \in N} y_j  - \frac{1}{2}  &\Vert v_0 \Vert ^ 2 \label{obj}\\
    y_j &\leq w_{ij}^2 - \frac{1}{2}\Vert v_{ij} \Vert ^ 2 + \langle v_0, v_{ij} \rangle \hspace{1.8cm} \forall j \in N, \forall i \in \mathcal{S}_j \label{first_set_constr}\\
    \langle v_{ij}, v_{i'k} \rangle &\leq 2 \: w_{ij} \: w_{ik} \: \mathds{1}_{\{i = i'\}} \hspace{3.1cm} \forall (i,j) \neq (i',k) \text{ with } j,k \in N. \label{second_set_constr}
\end{align}
\vspace{0.1cm}

\noindent Given any Nash equilibrium $x$, the goal is to use this dual program to construct a feasible solution with objective value at least $\rho \: C(x)$ for some $\rho \in [0,1]$. By weak duality, this would directly imply an upper bound of $1/\rho$ for the price of anarchy.
 
\paragraph{Correspondence between the SDP constraints and the Nash conditions.} The key insight is that the first set of constraints \eqref{first_set_constr} of the SDP has the same structure as that of the Nash equilibria inequalities \eqref{eq_nash_conditions_gen}. Our goal is to pick a fitting which will ensure that this set of constraints corresponds to (or is implied by) these equilibrium conditions. Fix a Nash equilibrium $x$ and let us write down what the Nash conditions imply for our toy example:

\[C_j(x) \leq w_{ij} \:  (\ell_i(x) + w_{ij}) = w_{ij}^2 + w_{ij} \: \ell_i(x) \qquad \forall j \in N, \forall i \in \mathcal{S}_j.\]
A natural way to achieve the desired correspondence is to have the following:
\begin{equation}
\label{eq_rough_fitting}
y_j \sim C_j(x) \quad , \quad w_{ij}^2 - \frac{1}{2} \Vert v_{ij} \Vert ^ 2 \sim w_{ij}^2 \quad, \quad \langle v_0, v_{ij} \rangle \sim w_{ij} \: \ell_i(x)
\end{equation}
where the $\sim$ notation indicates that both quantities are within a fixed constant (which should be the same for all three cases above) of each other. For local search algorithms, the Nash inequalities get replaced by local optima conditions. For online algorithms, these become inequalities satisfied by an online algorithm at every time step.

\paragraph{Picking the vector fitting.} Observe that the second correspondence above implies that $\Vert v_{ij} \Vert ^ 2 \sim w_{ij}^2$. The second set of SDP constraints \eqref{second_set_constr} tells us that for $i \neq i'$, one should have $\langle v_{ij}, v_{i'k} \rangle \leq 0$. We will in fact ensure tightness of this constraint by fitting such two vectors to be \emph{orthogonal}. A very natural candidate for the fitting of $v_{ij}$ in our example thus becomes the following choice:
\[v_{ij} \in \mathbb{R}^E \quad \text{ defined as  } \quad v_{ij}(e) = \alpha \: w_{ij} \: \mathds{1}_{\{i = e\}}\]
for some constant $\alpha \in [0, \sqrt{2}]$ to be determined. The upper bound on $\alpha$ follows again from the second set of SDP constraints \eqref{second_set_constr}, since we now get $\langle v_{ij}, v_{i'k} \rangle = \alpha^2 \: w_{ij} \: w_{i'k} \: \mathds{1}_{\{i = i'\}}$ under our fitting.

How should $v_0$ now be picked? There are two desirable properties to be satisfied: we want $\langle v_0, v_{ij} \rangle \sim w_{ij} \: \ell_i(x)$ as mentioned above, in addition to relating $\Vert v_0 \Vert ^2$ to the social cost $C(x)$, since it appears in the objective function of the SDP. A very natural candidate becomes the following:
\[v_0 \in \mathbb{R}^E \quad \text{ defined as  } \quad v_{0}(e) = \beta \: \ell_e(x)\]
where $\beta \geq 0$ is to be determined. Note that we now indeed get $\langle v_0, v_{ij} \rangle = \alpha \beta \: w_{ij} \: \ell_i(x)$ and $\Vert v_0 \Vert ^2 = \beta^2\: C(x)$, since \eqref{eq_social_cost_load_balancing} can be rewritten as $C(x) = \sum_{i \in E}\ell_i(x)^2$.

\paragraph{Optimizing the constants.} How should $\alpha$ and $\beta$ be picked? We have seen that $\alpha \in [0, \sqrt{2}]$ and $\beta \geq 0$. Observe that under our fitting, constraints \eqref{first_set_constr} now become $y_j \leq (1-\alpha^2/2) \: w_{ij}^2 + \alpha \beta \: w_{ij} \: \ell_i(x)$. Correspondence \eqref{eq_rough_fitting} then tells us to set $1-\alpha^2/2 = \alpha \beta$ and $y_j = \alpha \beta C_j(x)$. The objective value \eqref{obj} of the SDP then becomes $(\alpha \beta - \beta^2/2) \: C(x)$. Since we want to pick $\alpha$ and $\beta$ to maximize the dual objective in order to get the best possible bound on the price of anarchy, we would want to solve the following optimization problem:
\[\max \{\alpha \beta - \beta^2/2 : 1 - \alpha^2/2 = \alpha \beta, \alpha \in [0, \sqrt{2}], \beta \geq 0\}.\]
Solving this optimization problem would give a price of anarchy bound of $(3+\sqrt{5})/2$, which is tight in this setting by a lower bound construction of \cite{caragiannis2006tight}. At the high-level, this is the approach used to derive the results in this paper. We invite the reader to keep this intuition even for more complex games.
\subsection{Different inner product spaces}

 In order to construct a feasible solution to this SDP, one needs to construct vectors $v_0$ and $\{v_{ij}\}_{j \in N, i \in \mathcal{S}_j}$ living in a Euclidean space $\mathbb{R}^d$ for some $d > 0$, in addition to real values $\{y_j\}_{j \in N}$ such that both sets of constraints of the SDP are satisfied. Note that the inner product is the standard Euclidean one where, for given $f, g \in \mathbb{R}^d$, it is defined as
$\langle f, g \rangle := \sum_{i = 1}^d f_i g_i$. For some games, it will be more convenient to work in a different inner product space, as done in \cite{cole2011inner}. Let us fix a finite set $E$, where each $e \in E$ induces a finite set of positive real values $0 = \delta_0^{(e)} \leq \delta_1^{(e)} \leq \dots \leq \delta_n^{(e)}$. We define the following inner product space:
\begin{align*}
\mathcal{F}(E) := \left\{ f : E \times [0,\infty) \to \mathbb{R}_+ : f(e, t) = \sum_{j=1}^n \alpha_{ej} \: \mathds{1}_{\left\{\delta^{(e)}_{j-1} \: \leq t \:  \leq \: \delta^{(e)}_j\right\}}; \; \alpha_{ej} \in \mathbb{R} \quad \forall e \in E, \forall j \in [n] \right\}.
\end{align*}
In words, each element satisfies the fact that $f(e,\cdot)$ is a step-function with breakpoints at $\delta_1^{(e)} \leq \dots \leq \delta_n^{(e)}$ for every $e \in E$. A valid inner product for two $f,g \in \mathcal{F}(E)$ is then given by:
\begin{equation}
\label{eq_first_innerprod}
\langle f, g \rangle := \sum_{e \in E} \int_0^{\infty} f(e, t) \; g(e,t) \; dt.
\end{equation}
Another inner product space we will consider is the following. Let us fix $E$ to be a finite set and $K \in \mathbb{N}$. For any positive-definite matrix $M \in \mathbb{R}^{K \times K}$, we can consider the space $\mathcal{G}(E,M):= \mathbb{R}^{E \times [K]}$ where the inner product for two $f, g \in \mathcal{G}(E,M)$ is given by:
\begin{equation}
\label{eq_second_innerprod}
\langle f, g \rangle  := \sum_{e \in E} f(e, \cdot)^T \: M \: g(e, \cdot).
\end{equation}
We now argue that we can work in these spaces without loss of generality.
\begin{lemma}
\label{lemma_inner_product_space}
Any feasible dual fitting to (SDP-C) obtained in the inner product spaces $\mathcal{F}(E)$ and $\mathcal{G}(E,M)$ can be converted into a feasible dual fitting with the same objective value in a standard Euclidean space $\mathbb{R}^d$ for some $d > 0$ endowed with the standard inner product.
\end{lemma}
\begin{proof}
For both spaces, we argue that a collection of elements in it can be mapped to a collection of vectors in a standard Euclidean space while preserving all the pairwise inner products (and thus also preserving the norms). This then clearly implies the claim. 

We first show the statement for $\mathcal{F}(E)$. Let us denote the difference between two breakpoints as $\Delta_j^{(e)} := \delta_j^{(e)} - \delta_{j-1}^{(e)}$. For each element $f \in \mathcal{F}(E)$, define  $f' \in \mathbb{R}^{E \times [n]}$ as $f'(e,j) :=f\Big(e, \delta_{j-1}^{(e)} + \Delta_j^{(e)}/2\Big)  \sqrt{\Delta_j^{(e)}} $. By computing the integral of a step function, we clearly have that for $f, g \in \mathcal{F}(E)$,
\[\langle f, g \rangle = \sum_{e \in E} \sum_{j = 1}^{n} \Delta_j^{(e)} f\Big(e, \delta_{j-1}^{(e)} + \Delta_j^{(e)}/2\Big)  \; g\Big(e, \delta_{j-1}^{(e)} + \Delta_j^{(e)}/2\Big) = \langle f', g' \rangle.\]
Note that the last inner product is the standard Euclidean one, thus showing the claim for $\mathcal{F}(E)$.

We now show the claim for $\mathcal{G}(E,M)$. Let us write a rank one decomposition of the positive definite matrix $M = \sum_{j=1}^K u_j u_j^T$, which can for instance be done through the spectral decomposition. For each $f \in \mathcal{G}(E,M)$, we define a modified $f' \in \mathbb{R}^{E \times [K]}$ as $f'(e,j) := f (e, \cdot)^T u_j$. Clearly, we then have that,  for $f,g \in \mathcal{G}(E,M)$:
\[\langle f, g \rangle = \sum_{e \in E} \sum_{j = 1}^K f(e, \cdot)^T u_j \: u_j^T g(e, \cdot) = \sum_{e \in E} \sum_{j = 1}^K f'(e,j) g'(e,j) = \langle f', g' \rangle. \qedhere\]
\end{proof}

\section{Congestion games with coordination mechanisms}
\label{sec_cong_coor}

\subsection{Smith's Rule}
The first coordination mechanism we consider is Smith's Rule. If $x$ is a mixed assignment, each player first independently picks a strategy according to his/her distribution specified by $x$ to get an assignment. Once an assignment is set, each resource orders the players using it by increasing Smith ratios and processes them in that order. The cost incurred by a player $j$ on a resource $e$ that he/she is using is then $p_{ej} + \sum_{k \prec_e j} p_{ek} \: z_{ek}$. The total cost incurred by a player is the sum of the costs incurred on all the resources used. More formally, the completion time of player $j \in N$ under Smith's Rule is defined to be:
\begin{equation}
\label{eq_SR_formula}
C_j(x) = \sum_{i \in S_j} x_{ij} \sum_{e \in i}  \Big( p_{ej} +  \sum_{k \prec_e j} p_{ek} \: z_{ek} \Big).
\end{equation}
The outer sum only has one term for a binary assignment. For a mixed assignment $x$, the expression above is the expected cost, by the law of total probability and independence. The social cost is the sum of weighted completion times:
\begin{align}
\label{eq_soccost_selfrout}
C(x) := \sum_{j \in N} w_j \: C_j(x) = \sum_{e \in E} \sum_{j \in N} w_j \: p_{ej} \: z_{ej} + \sum_{e \in E} \sum_{j \in N, k \prec_e j} w_j \: p_{ek} \: z_{ek} \: z_{ej},
\end{align}
where the second equality follows by using the definition of $C_j(x)$ and changing the order of summation. Moreover, if $x$ is a Nash equilibrium, the following inequalities are satisfied:
\begin{equation}
\label{eq_Nash_selfrout_SR}
C_j(x) \leq \sum_{e \in i} \Big( p_{ej} + \sum_{k \prec_e j} p_{ek} z_{ek} \Big) \qquad \forall j \in N, \: \forall i \in \mathcal{S}_j.
\end{equation}
The dual semidefinite relaxation \eqref{dual_sdp} then becomes the following, we call it \emph{(SDP-SR)}. The computation of the cost matrix $C$ in this setting is shown in Appendix \ref{sec_comp_dual}.

\begin{align}
\label{sdp_sr}
    \max \sum_{j \in N} y_j  - \frac{1}{2}  &\Vert v_0 \Vert ^ 2 \\
    y_j &\leq \sum_{e \in i} w_j \: p_{ej} - \frac{1}{2}\Vert v_{ij} \Vert ^ 2 + \langle v_0, v_{ij} \rangle \qquad \qquad \forall j \in N, \forall i \in \mathcal{S}_j \label{sdp_sr_c1}\\
    \langle v_{ij}, v_{i'k} \rangle &\leq \sum_{e \in i \cap i'} w_{j} \: w_{k} \: \min{\{\delta_{ej}, \delta_{ek}\}}\hspace{2.2cm} \forall (i,j) \neq (i',k) \text{ with } j,k \in N. \label{sdp_sr_c2}
\end{align}
We note that, in order to bound the coordination ratio of a coordination mechanism, one needs to construct a feasible dual solution to this relaxation, since it gives a valid lower bound on the optimal solution. Indeed, once an assignment is fixed, the optimal ordering on every resource is to schedule the players according to Smith's Rule \cite{smith1956various}.

\begin{theorem}
\label{thm_Smith_Rule}
For any Nash equilibrium $x$  of the above congestion game, where each resource uses the Smith's Rule policy, there exists a feasible (SDP-SR) solution with value at least $1/4 \: C(x)$. This implies that the price of anarchy and the coordination ratio is at most $4$.
\end{theorem}
\begin{remark}
This bound is tight, with a matching lower bound given in \cite{correa2012efficiency} even for scheduling on restricted identical machines with unit processing times. 
\end{remark}
\begin{proof}
We assume that the SDP vectors live in the inner product space $\mathcal{F}(E)$. By Lemma \ref{lemma_inner_product_space}, this is without loss of generality.
Let us fix $\beta = 1/2$, we now state the dual fitting for \emph{(SDP-SR)}:
\begin{itemize}
\item $v_0(e,t) := \beta \sum_{k \in N} w_k \: z_{ek} \: \mathds{1}_{\{t \leq \delta_{ek}\}}$
\item $v_{ij}(e,t) := w_j \: \mathds{1}_{\{ e \in i\}} \: \mathds{1}_{\{ t \leq \delta_{ej}\}}  \hspace{2cm} \forall j \in N, \forall i \in \mathcal{S}_j$
\item $y_j := \beta \: w_j \: C_j(x) \hspace{3.9cm} \forall j \in N.$
\end{itemize}
Let us now compute the different inner products and norms that we need. For a job $j \in N$ and a strategy $i \in \mathcal{S}_j$, we have
\begin{align*}
\Vert v_{ij} \Vert ^ 2 = \sum_{e \in i} w_j ^ 2 \: \delta_{e j} = \sum_{e \in i} w_j \: p_{e j}.
\end{align*}
For $v_0$, we give an upper bound with respect to $C(x)$:
\begin{align}
\frac{1}{\beta^2}\: \Vert v_0 \Vert ^ 2 &= \sum_{e \in E} \sum_{j,k \in N} w_j \: w_k \: z_{ej} \: z_{ek} \int_0^{\infty} \mathds{1}_{\{t \leq \delta_{ej}\}} \mathds{1}_{\{t \leq \delta_{ek}\}} \: dt = \sum_{e \in E} \sum_{j,k \in N} w_j w_k z_{ej} z_{ek} \min\{\delta_{ej}, \delta_{ek}\} \nonumber\\
&= \sum_{e \in E} \sum_{j \in N} w_j ^ 2 \: z_{ej}^2 \: \delta_{ej} + 2 \sum_{e \in E} \sum_{j \in N, k \prec_e j}w_j w_k z_{ej} z_{ek} \delta_{ek} \nonumber\\
&=  \sum_{e\in E} \sum_{j \in N} w_j \: p_{ej} \: z_{ej}^2  + 2 \sum_{e \in E} \sum_{j \in N, k \prec_e j}w_j \: p_{ek} \: z_{ej} \: z_{ek} \nonumber\\
&\leq 2 \: C(x).  \label{eq_v0}
\end{align}
The last equality uses the definition of the Smith Ratio $\delta_{ej} = p_{ej}/w_j$, whereas the last inequality follows from the fact that $z_{ej}^2 \leq z_{ej}$ (since $z_{ej} \in [0,1]$) as well as the definition of the social cost \eqref{eq_soccost_selfrout}. In addition, for any $(i,j) \neq (i',k)$ with $j,k \in N$, we have
\begin{equation}
\label{eq_inner_prod_jk}
\langle v_{ij} ,v_{i'k} \rangle = \sum_{e \in E} \: w_j \: w_k \: \mathds{1}_{\{ e \in i\}} \: \mathds{1}_{\{ e \in i'\}}\: \int_0^{\infty} \mathds{1}_{\{t \leq \delta_{ej}\}} \mathds{1}_{\{t \leq \delta_{ek}\}} \: dt = \sum_{e \in i \cap i'} w_{j} \: w_{k} \: \min{\{\delta_{ej}, \delta_{ek}\}}
\end{equation}
and observe that this tightly satisfies the second set of SDP constraints \eqref{sdp_sr_c2}. Finally,
\[\langle v_0 ,v_{ij} \rangle = \beta \sum_{e \in i} \sum_{k \in N} w_j \: w_k \: z_{ek} \:  \min\{\delta_{ej}, \delta_{ek}\}.\]
Let us now check that this is a feasible solution to \emph{(SDP-SR)}. The second set of constraints is satisfied due to \eqref{eq_inner_prod_jk}. The first set of constraints \eqref{sdp_sr_c1} under the above fitting becomes:
\begin{align*}
y_j &\leq \sum_{e \in i} w_j \: p_{ej} - \frac{1}{2}\Vert v_{ij} \Vert ^ 2 + \langle v_0, v_{ij} \rangle \\
&\iff \beta \: w_j \: C_j(x) \leq \frac{1}{2} \sum_{e \in i} w_j \: p_{ej} + \beta \sum_{e \in i} \sum_{k \in N} w_j \: w_k \: z_{ek} \:  \min\{\delta_{ej}, \delta_{ek}\} \\
&\iff C_j(x) \leq \sum_{e \in i} \Big(p_{ej} +  \sum_{k \in N} w_k \: z_{ek} \: \min\{\delta_{ej}, \delta_{ek}\} \Big)\\
&\iff C_j(x) \leq \sum_{e \in i} \Big( p_{ej} +  \sum_{k \prec_e j} p_{ek} \: z_{ek} +  \sum_{k \succeq_e j} w_k \: z_{ek} \: \delta_{ej} \Big).
\end{align*}
We have simplified both sides by $\beta \: w_j = 1/2 \: w_j $ in the third line, which holds by definition of $\beta := 1/2$. We have also used the definition of the Smith ratio $\delta_{ek} = p_{ek}/w_k$ in the last line. This set of constraints is now clearly satisfied by the Nash conditions \eqref{eq_Nash_selfrout_SR}. The objective function of this SDP can now be lower bounded using \eqref{eq_v0}:
\begin{align*}
\sum_{j\in N} y_j  - \frac{1}{2} &\Vert v_0 \Vert ^ 2 \geq \beta \sum_{j\in N} w_j \: C_j(x) - \beta^2 C(x) = \left(\beta - \beta^2 \right) C(x) = \frac{1}{4} \: C(x). \qedhere
\end{align*}
\end{proof}

\subsection{The Proportional Sharing policy}
In this section, we consider a preemptive policy for every resource named \emph{Proportional Sharing}. Once an assignment is fixed, each resource splits its processing capacity among the uncompleted jobs proportionally to their weights. Given an assignment $x$, the completion time of player $j \in N$ is defined to be:
\[C_j(x) = \sum_{i \in S_j} x_{ij} \sum_{e \in i}  \Big( p_{ej} +  \sum_{k \prec_e j} p_{ek} \: z_{ek} + \sum_{k \succ_e j} w_k \: z_{ek} \: \delta_{ej} \Big). \]
For this policy, it is in fact more intuitive to understand the definition by looking at the weighted completion time:
\[w_j C_j(x) = \sum_{i \in S_j} x_{ij} \sum_{e \in i}  \Big( w_j p_{ej} +  \sum_{k \in N \setminus{\{j\}}} w_j w_k \min \{\delta_{ej}, \delta_{ek} \} \: z_{ek} \Big).\]
The social cost is the sum of weighted completion times:
\begin{align}
\label{eq_soccost_selfrout_ps}
C(x) := \sum_{j \in N} w_j \: C_j(x) = \sum_{e \in E} \sum_{j \in N} w_j \: p_{ej} \: z_{ej} + 2 \sum_{e \in E} \sum_{j \in N, k \prec_e j} w_j \: p_{ek} \: z_{ek} \: z_{ej}.
\end{align}
Observe that there is now a factor 2 in front of the second term if one compares it to the \emph{Smith Rule} policy. Moreover, if $x$ is a Nash equilibrium, the following inequalities are satisfied:
\begin{equation}
\label{eq_Nash_selfrout}
C_j(x) \leq \sum_{e \in i} \Big( p_{ej} + \sum_{k \prec_e j} p_{ek} z_{ek} + \sum_{k \succ_e j} w_k \: z_{ek} \: \delta_{ej}\Big) \qquad \forall j \in N, \: \forall i \in \mathcal{S}_j.
\end{equation}
We first need a small lemma about two parameters that will play a key role in the dual fitting. The first property will ensure feasibility of the solution, whereas the second one will be the constant in front of the objective function.
\begin{lemma}
\label{lemma_alpha_beta_cong}
Let $\alpha, \beta \geq 0$ be defined as $\alpha^2 := 2/\sqrt{5}$ and $\beta := 1/\alpha - \alpha/2$. The following two properties hold:
\begin{itemize}
\item $1 - \alpha^2/2 = \alpha \beta$
\item $\alpha \beta - \beta^2/2 = 2/(3+\sqrt{5})$
\end{itemize}
\end{lemma}
\begin{proof}
The first property is immediate by definition of $\beta$. For the second property, we get
\begin{align*}
\alpha \beta - \frac{\beta^2}{2} = 1 - \frac{\alpha^2}{2} - \frac{1}{2}\left(\frac{1}{\alpha} - \frac{\alpha}{2}\right)^2 = \frac{3}{2} - \frac{5\alpha^2}{8} - \frac{1}{2\alpha^2} = \frac{3}{2} - \frac{5}{4 \sqrt{5}} - \frac{\sqrt{5}}{4} = \frac{2}{3 + \sqrt{5}}.
\end{align*}
\qedhere
\end{proof}
\begin{theorem}
For any Nash equilibrium $x$  of the above congestion game, where each resource uses the Proportional Sharing policy, there exists a feasible (SDP-SR) solution with value at least $2/(3 + \sqrt{5}) \: C(x)$, implying that the coordination ratio is at most $(3+\sqrt{5})/2$.
\end{theorem}
\begin{remark}
This bound is tight, with a matching lower bound given in \cite{caragiannis2006tight} even for the price of anarchy of the game.
\end{remark}
\begin{proof}
The proof is very similar to the one of Theorem \ref{thm_Smith_Rule}, but with the modified constants $\alpha^2 := 2/\sqrt{5}$ and $\beta := 1/\alpha - \alpha/2$ stated in Lemma \ref{lemma_alpha_beta_cong}. We now state the dual fitting.
\begin{itemize}
\item $v_0(e,t) := \beta \sum_{k \in N} w_k \: z_{ek} \: \mathds{1}_{\{t \leq \delta_{ek}\}}$
\item $v_{ij}(e,t) := \alpha \; w_j \: \mathds{1}_{\{ e \in i\}} \: \mathds{1}_{\{ t \leq \delta_{ej}\}}  \hspace{2cm} \forall j \in N, \forall i \in \mathcal{S}_j$
\item $y_j := \alpha \beta \: w_j \: C_j(x) \hspace{3.95cm} \forall j \in N.$
\end{itemize}
Using the same computations as in Theorem \ref{thm_Smith_Rule}, we compute the different inner products and norms that we need.
\begin{itemize}
\item $\Vert v_0 \Vert ^ 2 = \beta^2 \left(\sum_{e\in E} \sum_{j \in N} w_j \: p_{ej} \: z_{ej}^2  + 2 \sum_{e \in E} \sum_{j \in N, k \prec_e j}w_j \: p_{ek} \: z_{ej} \: z_{ek}\right) \leq \beta^2 C(x)$
\item $\Vert v_{ij} \Vert ^ 2 = \alpha^2 \sum_{e \in i} w_j \: p_{e j}$
\item $\langle v_0, v_{ij} \rangle = \alpha \beta \sum_{e \in i} \sum_{k \in N} w_j \: w_k \: z_{ek} \:  \min\{\delta_{ej}, \delta_{ek}\}$
\item $\langle v_{ij}, v_{i'k} \rangle = \alpha^2 \; \sum_{e \in i \cap i'} w_{j} \: w_{k} \: \min{\{\delta_{ej}, \delta_{ek}\}}$
\end{itemize}
The main difference with \emph{Smith's Rule} which allows us to get an improved bound is the fact that the upper bound on the squared norm of $v_0$ is a factor $2$ stronger in this case (see {\eqref{eq_v0}}), due to the new definition of the social cost $C(x)$ given in \eqref{eq_soccost_selfrout_ps}. To see that this solution is feasible, note that the second set of SDP constraints \eqref{sdp_sr_c2} is satisfied due to the last computation above and the fact that $\alpha^2 \leq 1$. The first set of constraints \eqref{sdp_sr_c1} under the above fitting reads:
\begin{align*}
y_j &\leq \sum_{e \in i} w_j \: p_{ej} - \frac{1}{2}\Vert v_{ij} \Vert ^ 2 + \langle v_0, v_{ij} \rangle \\
&\iff \alpha \beta \: w_j \: C_j(x) \leq \left(1 - \frac{\alpha^2}{2}\right) \sum_{e \in i} w_j \: p_{ej} + \alpha \beta \sum_{e \in i} \sum_{k \in N} w_j \: w_k \: z_{ek} \:  \min\{\delta_{ej}, \delta_{ek}\} \\
&\iff C_j(x) \leq \sum_{e \in i} \Big(p_{ej} +  \sum_{k \in N} w_k \: z_{ek} \: \min\{\delta_{ej}, \delta_{ek}\} \Big)\\
&\iff C_j(x) \leq \sum_{e \in i} \Big( p_{ej} +  \sum_{k \prec_e j} p_{ek} \: z_{ek} +  \sum_{k \succeq_e j} w_k \: z_{ek} \: \delta_{ej} \Big).
\end{align*}
The third equivalence follows from the first property of Lemma \ref{lemma_alpha_beta_cong}. We see that this is satisfied due to the Nash conditions \eqref{eq_Nash_selfrout}.
The objective value of the solution can now be lower bounded as follows:
\[\sum_{j \in N} y_j  - \frac{1}{2}  \Vert v_0 \Vert ^ 2 \geq  \alpha \beta \sum_{j \in N}w_j C_j(x) - \frac{\beta^2}{2} \: C(x) = \left(\alpha \beta - \frac{\beta^2}{2}\right) C(x) = \frac{2}{3 + \sqrt{5}} \: C(x)\]
where the last equality follows by the second property of Lemma \ref{lemma_alpha_beta_cong}.
\end{proof}

\subsection{The Rand policy}
In this section, we consider a randomized policy named \emph{Rand}. If $x$ is a mixed assignment, each player first independently picks a strategy according to his/her distribution specified by $x$. We denote by $N(e) \subseteq N$ the (possibly random) subset of players using resource $e \in E$. Each resource then orders the players using it randomly in a way ensuring that for any pair $j,k \in N(e)$, the probability that $j$ comes after $k$ in the ordering is exactly equal to 
$\delta_{ej}/(\delta_{ej} + \delta_{ek})$. Such a distribution can be achieved by sampling one player $j \in N(e)$ at random with probability $\delta_{ej}/\sum_{k \in N(e)}\delta_{ek}$, putting that player at the end of the ordering, and repeating this process. The expected completion time of every player is thus given by:
\[C_j(x) = \sum_{i \in S_j} x_{ij} \sum_{e \in i}  \Big( p_{ej} +  \sum_{k \neq j} \frac{\delta_{ej}}{\delta_{ej} + \delta_{ek}} \: p_{ek} \: z_{ek} \Big). \]
The social cost is the sum of weighted completion times:
\begin{align}
\label{eq_soccost_rand}
C(x) := \sum_{j \in N} w_j \: C_j(x) = \sum_{e \in E} \sum_{j \in N} w_j \: p_{ej} \: z_{ej} + \sum_{e \in E} \sum_{j \in N, k \neq j} \frac{\delta_{ej} \delta_{ek}}{\delta_{ej} + \delta_{ek}} \: w_j w_k \: z_{ej} z_{ek}.
\end{align}
Moreover, if $x$ is a Nash equilibrium, the following inequalities are satisfied: 
\begin{equation}
\label{eq_Nash_rand}
C_j(x) \leq \sum_{e \in i} \Big( p_{ej} + \sum_{k \neq j} \frac{\delta_{ej}}{\delta_{ej} + \delta_{ek}} \: p_{ek} \: z_{ek} \Big) \qquad \qquad \forall j \in N, i \in \mathcal{S}_j.
\end{equation}
We now state a small lemma about some constants that will be important for the fitting. The first property will ensure that our dual fitting is feasible, whereas the second property will be the constant in front of the objective value of our SDP solution, thus determining the coordination ratio.
\begin{lemma}
\label{lemma_alpha_beta_rand_2}
Let $\alpha, \beta \geq 0$ be defined as $\alpha := 1$ and $\beta := 3/4$. The following two properties hold:
\begin{itemize}
\item $1 - \alpha^2/4 = \alpha \beta$
\item $\alpha \beta - \beta^2/2 = 15/32$
\end{itemize}
\end{lemma}
\begin{proof} The proof is immediate.
\end{proof}
\begin{theorem}
\label{thm_dual_fitting_rand}
For any instance of the above congestion game under the Rand policy, and for any Nash equilibrium $x$, there exists a feasible (SDP-SR) solution with value at least $15/32 \; C(x)$. This implies that the coordination ratio is at most $32/15 \approx 2.133$.
\end{theorem}
\begin{proof}
For simplicity of presentation, let us assume without loss of generality that the processing times are scaled such that the Smith ratios $\delta_{ej} = p_{ej}/w_j$ with $p_{ej} < \infty$ are all integral.  Moreover, let us take $K \in \mathbb{N}$ large enough such that $\delta_{ej} \leq K$ for all pairs $(e,j) \in E \times N$ such that $p_{ej} < \infty$. 
Consider the matrix $M \in \mathbb{R}^{K \times K}$ given by 
\[M_{r,s} := \frac{r \; s}{r + s} \qquad \forall r, s \in \{1, \dots, K\}.\]
A key insight shown in \cite{cole2011inner} is that this matrix is positive-definite. By Lemma \ref{lemma_inner_product_space}, we can thus assume that the SDP vectors live in the space $\mathcal{G}(E,M)$. Let $\alpha, \beta$ be defined as in Lemma \ref{lemma_alpha_beta_rand_2}, we now state the dual fitting:
\begin{itemize}
\item $v_0(e, r) := \beta \: \sum_{k \in N} w_k \; z_{ek} \; \mathds{1}_{\{\delta_{ek} = r\}}$
\item $v_{ij}(e, r) := \alpha \: w_j \; \mathds{1}_{\{e \in i\}} \; \mathds{1}_{\{\delta_{ej} = r\}} \qquad  \qquad \qquad \forall j \in N, i \in \mathcal{S}_j$
\item $y_j := \alpha \beta \: w_j \: C_j(x) \qquad \qquad \qquad \qquad \qquad \qquad \forall j \in N.$
\end{itemize}
Let us now compute the different inner products and norms that we need. For every $j \in N, i \in \mathcal{S}_j$:
\begin{align*}
\frac{1}{\alpha^2} \; \Vert v_{ij}\Vert ^ 2 = \sum_{e \in i} M_{\{\delta_{ej},\delta_{ej}\}} w_j^2 = \sum_{e \in i}\frac{\delta_{ej}}{2} \: w_{j}^2 = \frac{1}{2} \: \sum_{e \in i} w_j \: p_{ej}.
\end{align*}
For the squared norm of $v_0$, we give an upper bound with respect to $C(x)$:
\begin{align}
\frac{1}{\beta^2}\: \Vert v_0 \Vert ^ 2 &= \sum_{e \in E} \sum_{r,s = 1}^K M_{r,s} \: v_0(e,r) \: v_0(e,s) = \sum_{e \in E} \sum_{j,k \in N} w_j \: w_k \: z_{ej} \: z_{ek} \: M_{\{\delta_{ej},\delta_{ek}\}} \nonumber \\
&= \sum_{e \in E} \sum_{j,k \in N} \frac{\delta_{ej} \delta_{ek}}{\delta_{ej} + \delta_{ek}}w_j \: w_k \: z_{ej} \: z_{ek} \leq C(x)
\end{align}
where the last inequality holds by \eqref{eq_soccost_rand} and $z_{ej}^2 \leq z_{ej}$.
For any pair $(i,j) \neq (i',k)$ with $j,k \geq 1$:
\begin{align}
\label{eq_rand_sdp_cons}
\frac{1}{\alpha^2} \; \langle v_{ij}, v_{i'k} \rangle =  \sum_{e \in i \cap i'} M_{\{\delta_{ej},\delta_{ek}\}} w_j \: w_k = \sum_{e \in i \cap i' } w_j \: w_k \: \frac{\delta_{ej} \: \delta_{ek}}{\delta_{ej} + \delta_{ek}}.
\end{align}
Finally, we have:
\begin{align*}
\frac{1}{\alpha \beta} \: \langle v_0, v_{ij} \rangle &= \sum_{e \in i} \sum_{k \in N} w_j \: w_k \: z_{ek} \: M_{\{\delta_{ej},\delta_{ek}\}} = \sum_{e \in i} \sum_{k \in N} \frac{\delta_{ej} \delta_{ek}}{\delta_{ej} +\delta_{ek}} w_j \: w_k \: z_{ek}\\ &= w_j \sum_{e \in i} \sum_{k \in N} \frac{\delta_{ej}}{\delta_{ej} +\delta_{ek}} \: p_{ek} \: z_{ek}
\end{align*}
where the last equality follows by plugging in the definition of $\delta_{ek} = p_{ek}/w_k$.

Let us now check that this solution is indeed feasible for \emph{(SDP-SR)}. The second set of constraints \eqref{sdp_sr_c2} is satisfied due to \eqref{eq_rand_sdp_cons}, the fact that $\alpha = 1$, as well as observing that $rs/(r+s) \leq \min\{r,s\}$ for all $r,s \geq 0$. The first set of constraints \eqref{sdp_sr_c1} under the above fitting gives:
\begin{align*}
y_j &\leq \sum_{e \in i} w_j \: p_{ej} - \frac{1}{2}\Vert v_{ij} \Vert ^ 2 + \langle v_0, v_{ij} \rangle \\
&\iff \alpha \beta \: w_j \: C_j(x) \leq \left(1 - \frac{\alpha^2}{4} \right) \sum_{e \in i} w_j \: p_{ej} + \alpha \beta \: w_j \sum_{e \in i} \sum_{k \in N} \frac{\delta_{ej}}{\delta_{ej} +\delta_{ek}} \: p_{ek} \: z_{ek} \\
&\iff C_j(x) \leq \sum_{e \in i} \Big(p_{ej} +  \sum_{k \in N} \frac{\delta_{ej}}{\delta_{ej} +\delta_{ek}} \: p_{ek} \: z_{ek} \Big).
\end{align*}
We have simplified both sides by $\alpha \beta \: w_j = (1 - \alpha^2/4) w_j$ in the last equivalence, which holds by the first property of Lemma \ref{lemma_alpha_beta_rand_2}. These inequalities are now clearly satisfied by the Nash conditions \eqref{eq_Nash_rand}, implying that our fitted solution is in fact feasible. The objective value of our solution can be lower bounded as:
\[\sum_{j\in N} y_j  - \frac{1}{2} \Vert v_0 \Vert ^ 2 \geq \alpha \beta \sum_{j \in N} w_j C_j(x) - \frac{\beta^2}{2}C(x) = \left( \alpha \beta - \frac{\beta^2}{2}\right)C(x) = \frac{15}{32}\: C(x)\]
where the last equality follows from the second property of Lemma \ref{lemma_alpha_beta_rand_2}.
\end{proof}

We now show that this bound can be improved if we consider the natural special case where the processing time of each player is proportional to its weight for every resource. This means that every resource has a real-value $\lambda_e \geq 0$, and the processing time of every player satisfies $p_{ej} \in \{\lambda_e w_j, \infty\}$ for every $e \in E, j \in N$. Observe that this means that the Smith ratios are uniform for the jobs assigned to a resource: $\delta_{ej} = p_{ej}/w_j = \lambda_e$. The only difference with respect to the previous proof will be a change of constants $\alpha, \beta$.

\begin{lemma}
\label{lemma_alpha_beta_rand_mod}
Let $\alpha, \beta \geq 0$ be defined as $\alpha := 2/\sqrt{3}$ and $\beta := 1/\sqrt{3}$. The following two properties hold:
\begin{itemize}
\item $1 - \alpha^2/4 = \alpha \beta$
\item $\alpha \beta - \beta^2/2 = 1/2$
\end{itemize}
\end{lemma}
\begin{proof}
The proof is immediate.
\end{proof}

\begin{theorem}
If the Smith ratios are uniform for every resource, for any instance of the above game and any Nash equilibrium $x$, there exists a feasible (SDP-SR) solution with value at least $1/2 \: C(x)$. This implies that the coordination ratio of the game is at most $2$.
\end{theorem}
\begin{proof}
Let $\alpha, \beta$ be as in Lemma \ref{lemma_alpha_beta_rand_mod}. The only part of the proof of Theorem \ref{thm_dual_fitting_rand} which breaks down under these new constants is the fact that the second set of constraints \eqref{sdp_sr_c2} of the SDP might be violated, since we now have $\alpha^2 = 4/3 > 1$. Indeed \eqref{eq_rand_sdp_cons} states that:
\[\langle v_{ij}, v_{i'k} \rangle = \alpha ^2 \sum_{e \in i \cap i' } w_j \: w_k \: \frac{\delta_{ej} \: \delta_{ek}}{\delta_{ej} + \delta_{ek}}. \]
The proof of Theorem \ref{thm_dual_fitting_rand} used the easy observation that $rs/(r+s) \leq \min\{r,s\}$ for every $r,s \geq 0$ to argue feasibility of the solution. Observe that this bound is very close to tight when $s \gg r$ (or vice versa). In the case of uniform Smith ratios, we can get an improved bound since $\delta_{ej} = \delta_{ek} = \lambda_e$:
\[\langle v_{ij}, v_{i'k} \rangle = \alpha ^2 \sum_{e \in i \cap i' } w_j \: w_k \: \frac{\lambda_e}{2} \leq \sum_{e \in i \cap i' } w_j \: w_k \: \lambda_e =  \sum_{e \in i \cap i' } w_j \: w_k \: \min\{\delta_{ej}, \delta_{ek}\}\]
where the inequality follows since $\alpha^2/2 = 4/6 \leq 1$.
By the second property of Lemma \ref{lemma_alpha_beta_rand_mod}, the objective value can now be lower bounded as
\[\sum_{j\in N} y_j  - \frac{1}{2} \Vert v_0 \Vert ^ 2 \geq \alpha \beta \sum_{j \in N} w_j C_j(x) - \frac{\beta^2}{2}C(x) = \left( \alpha \beta - \frac{\beta^2}{2}\right)C(x) = \frac{1}{2}\: C(x). \qedhere\]
\end{proof}

We now show that this bound of $2$ can also be attained for arbitrary instances if we consider the price of anarchy of the game, rather than the coordination ratio, meaning that we now compare against the optimal solution under the \emph{Rand} policy. More precisely, we compare against the best possible assignment $x$, whose expected cost is measured if every resource uses the \emph{Rand} policy to process the players. Note that this cost is always higher than the cost if every resource were to use \emph{Smith's Rule}. In that case, a relaxation giving a valid lower bound on the social optimum is the following, we call it \emph{(SDP-RAND)}. The computation of the cost matrix $C$ to plug-in in \eqref{dual_sdp} in this setting is once again left to Appendix \ref{sec_comp_dual}.

\begin{align*}
    \max \sum_{j \in N} y_j  - \frac{1}{2}  &\Vert v_0 \Vert ^ 2 \\
    y_j &\leq \sum_{e \in i} w_j \: p_{ej} - \frac{1}{2}\Vert v_{ij} \Vert ^ 2 + \langle v_0, v_{ij} \rangle \qquad \qquad \forall j \in N, \forall i \in \mathcal{S}_j\\
    \langle v_{ij}, v_{i'k} \rangle &\leq 2 \sum_{e \in i \cap i'} w_{j} \: w_{k} \: \frac{\delta_{ej} \: \delta_{ek}}{\delta_{ej} + \delta_{ek}}\hspace{2.2cm} \forall (i,j) \neq (i',k) \text{ with } j,k \in N.
\end{align*}

\begin{theorem}
For any instance of the above game under the Rand policy, and for any Nash equilibrium $x$, there exists a feasible (SDP-RAND) solution with value at least $1/2 \: C(x)$. This implies that the price of anarchy of the game is at most $2$.
\end{theorem}
\begin{proof}
The proof is identical to the one of Theorem \ref{thm_dual_fitting_rand}, but with the modified constants $\alpha, \beta$ stated in Lemma \ref{lemma_alpha_beta_rand_mod}. This new choice of constants is not valid for \emph{(SDP-SR)}, due to the fact that $\alpha^2 >1$. Indeed, equation \eqref{eq_rand_sdp_cons} means that the second set of constraints \eqref{sdp_sr_c2} of \emph{(SDP-SR)} might now be violated. However, the second set of constraints of \emph{(SDP-RAND)} is always satisfied, since $\alpha^2 \leq 2$. The objective function guarantee follows from the second property of Lemma \ref{lemma_alpha_beta_rand_mod}.
\end{proof}

\section{Weighted affine congestion games}
\label{sec_weig_cong}
In this section, we consider the classic weighted affine congestion game. The price of anarchy of this game was settled in \cite{awerbuch2005price, christodoulou2005price} with a tight bound of $(3 + \sqrt{5})/2$ and this bound can also be obtained through a dual fitting argument on a convex program \cite{kulkarni2014robust}. We show here how to recover this bound in a simple way through our approach. For simplicity of presentation, we assume in this section that the Nash equilibria considered are pure, extensions to more general equilibrium notions can be found in Appendix \ref{section_robust}.

The setting is the following. There is a set $N$ of players and a set $E$ of resources. The strategy set for each player $j \in N$ is denoted by $\mathcal{S}_j \subseteq 2^E$ and is a collection of subsets of resources. Let us also assume that we have unrelated weights $w_{ej} \geq 0$ for every $j \in N, e \in E$. Given a strategy profile $x$, the \emph{load} of a resource is given by:
\[\ell_e(x) := \sum_{j \in N} w_{ej} \sum_{i \in \mathcal{S}_j : \: e \in i} x_{ij}.\]
The cost incurred by a player $j$ for a pure assignment $x$ is then given by
\[C_j(x) := \sum_{i \in \mathcal{S}_j} x_{ij} \sum_{e \in i} w_{ej} \; (a_e \: \ell_e(x) + b_e)\]
where $a_e, b_e \in \mathbb{R}_{\geq 0}$ for every $e \in E$. The social cost then becomes:
\begin{equation}
\label{eq_scost_cong}
C(x) := \sum_{j \in N} C_j(x) = \sum_{e \in E} a_e \: \ell_e(x)^2 + b_e \:  \ell_e(x)
\end{equation}
where the last equality holds by changing the order of summation and using the definition of $\ell_e(x)$. 

The Nash equilibrium conditions imply the following constraints for every $j \in N, i \in \mathcal{S}_j$:
\begin{align}
\label{eq_nash_cong}
C_j(x) \leq \sum_{e \in i} w_{ej} \; \Big(a_e \: (\ell_e(x) + w_{ej}) + b_e\Big) =  \sum_{e \in i} w_{ej}(a_e \: w_{ej} + b_e) + \sum_{e \in i} w_{ej} \: a_e \: \ell_e(x).
\end{align}
Indeed, if a player $j \in N$ decides to switch to a strategy $i \in \mathcal{S}_j$, then the load on every resource $e \in i$ can go up by at most $w_{ej}$. The semidefinite relaxation \eqref{dual_sdp} in this special case becomes the following, we call it \emph{(SDP-CG)}.
\begin{align*}
    \max \sum_{j \in N} y_j  - \frac{1}{2}  &\Vert v_0 \Vert ^ 2 \\
    y_j &\leq \sum_{e \in i} w_{ej}(a_e \: w_{ej} + b_e) - \frac{1}{2}\Vert v_{ij} \Vert ^ 2 + \langle v_0, v_{ij} \rangle \qquad \forall j \in N, \forall i \in \mathcal{S}_j\\
    \langle v_{ij}, v_{i'k} \rangle &\leq 2 \sum_{e \in i \cap i'} a_e \: w_{ej} \: w_{ek} \hspace{4cm} \forall (i,j) \neq (i',k) \text{ with } j,k \in N.
    \end{align*}
\begin{theorem}
For any instance of the above game, and any Nash equilibrium $x$, there exists a feasible (SDP-CG) solution with objective value at least $2/(3+\sqrt{5})$C(x).
\end{theorem}
\begin{proof}
The vectors of the SDP will live in the space $\mathbb{R}^E$. Let $\alpha, \beta \geq 0$ be defined as in Lemma \ref{lemma_alpha_beta_cong}. We now state the dual fitting:
\begin{itemize}
\item $v_0(e) := \beta \: \sqrt{a_e} \: \ell_e(x)$
\item $v_{ij}(e) := \alpha \: \sqrt{a_e} \: w_{ej} \: \mathds{1}_{\{e \in i\}}$ \qquad \qquad \qquad  $\forall j \in N, i \in \mathcal{S}_j$
\item $y_j := \alpha \beta \; C_j(x) \hspace{3.9cm} \forall j \in N.$
\end{itemize}
Let us now compute the different inner products and norms that we need.
\begin{itemize}
\item $\Vert v_0 \Vert ^ 2 = \beta^2 \; \sum_{e \in E} \: a_e \: \ell_e(x)^2 \leq \beta^2 \: C(x)$
\item $\Vert v_{ij} \Vert ^ 2 = \alpha^2 \; \sum_{e \in i} a_e \: w_{ej}^2$
\item $\langle v_0, v_{ij} \rangle = \alpha \beta \; \sum_{e \in i} a_e \: w_{ej} \: \ell_e(x)$
\item $\langle v_{ij}, v_{i'k} \rangle = \alpha^2 \; \sum_{e \in i \cap i'} a_e \: w_{ej} \: w_{ek}$
\end{itemize}
Let us now check feasibility of the solution. The second set of constraints is satisfied by the fourth computation above and the fact that $\alpha^2 = 2/\sqrt{5} \leq 2$. The first set of constraints is satisfied due to the Nash conditions \eqref{eq_nash_cong}. Indeed, under the above fitting, for every $j \in N, i \in \mathcal{S}_j$, the first set of SDP constraints reads:
\begin{align*}
\alpha \beta \: C_j(x) &\leq (1 - \alpha^2/2) \sum_{e \in i} a_e \: w_{ej}^2 + \sum_{e \in i} w_{ej} \: b_e + \alpha \beta \; \sum_{e \in i} a_e \: w_{ej} \: \ell_e(x).
\end{align*}
If there was a factor of $(1 - \alpha^2/2) \leq 1$ multiplying the term $\sum_{e \in i} w_{ej} \: b_e$, then this would be equivalent to \eqref{eq_nash_cong} because of the first condition of Lemma \ref{lemma_alpha_beta_cong}. Not having this term only increases the right-hand side and thus ensures that this set of constraints is satisfied, implying that the SDP solution is feasible. The objective function can now be lower bounded as:
\[\sum_{j \in N} y_j  - \frac{1}{2}  \Vert v_0 \Vert ^ 2 \geq  \alpha \beta \sum_{j \in N}C_j(x) - \frac{\beta^2}{2} \: C(x) = \left(\alpha \beta - \frac{\beta^2}{2}\right) C(x) = \frac{2}{3 + \sqrt{5}} \: C(x)\]
where the last equality follows by the second property of Lemma \ref{lemma_alpha_beta_cong}.
\end{proof}

\section{Analyzing local search algorithms for scheduling}
\label{sec_local_opt}
We now show that this approach can also be useful to bound the approximation ratio of local search algorithms. We focus on the $R || \sum w_j C_j$ scheduling problem, for which the \emph{(SDP-SR)} relaxation \eqref{sdp_sr} becomes the following:

\begin{align*}
    \max \sum_{j \in J} y_j  - \frac{1}{2}  &\Vert v_0 \Vert ^ 2 \\
    y_j &\leq w_j p_{ij} - \frac{1}{2}\Vert v_{ij} \Vert ^ 2 + \langle v_0, v_{ij} \rangle \hspace{2.2cm} \forall j \in J, \forall i \in \mathcal{S}_j\\
    \langle v_{ij}, v_{i'k} \rangle &\leq  w_j \: w_k \: \min\{\delta_{ij}, \delta_{ik}\} \:  \mathds{1}_{\{i = i'\}} \hspace{2cm} \forall (i, j) \neq (i',k) \text{ with } j,k \in J.
\end{align*}

\noindent Given an assignment $x \in \{0,1\}^{M \times J}$, the completion time of every job $j \in J$ is:
\[C_j(x) = \sum_{i \in M} x_{ij} \Big(p_{ij} + \sum_{k \prec_i j} p_{ik} x_{ik}\Big).\]
Let us also define the following quantity for every $j \in J$:
\begin{equation}
D_j(x) = \sum_{i \in M} \sum_{k \succ_i j} w_k \: p_{ij} \: x_{ij} \: x_{ik}
\end{equation}
and let us denote the weighted sum of processing times as:
\begin{equation}
\eta(x) = \sum_{i \in M} \sum_{j \in J} w_j \: p_{ij} \: x_{ij}.
\end{equation}
The total cost can then be written in the following ways:
\begin{align}
C(x) = \sum_{j \in J} w_j C_j(x) = \eta(x) + \sum_{i \in M} \sum_{j \in J} \sum_{k \prec_i j} w_j \: p_{ik} \: x_{ij} \: x_{ik} \label{cost_1}\\
C(x) = \eta(x) + \sum_{j \in J} D_j(x) = \eta(x) + \sum_{i \in M} \sum_{j \in J} \sum_{k \succ_i j} w_k \: p_{ij} \: x_{ij} \: x_{ik}. \label{cost_2}
 \end{align}
 
 \subsection{A simple and natural local search algorithm}
A natural and simple local search algorithm for this problem is to move a job from one machine to another if that improves the objective function. If such an improvement is not possible, then a local optimum $x \in \{0,1\}^{M \times J}$ has been reached. Such a local optimum is called a $\emph{JumpOpt}$ in \cite{correa2022performance}, and it is shown that the local optimality implies the following constraints. We provide a proof for the sake of completeness.
\begin{lemma}
\label{lemma_local_opt}
For any local optimum JumpOpt solution $x$ of the scheduling problem $R || \sum w_j C_j$, the following constraints are satisfied:
\begin{align*}
w_j C_j(x) + D_j(x) \leq w_j \: p_{ij} + \sum_{k \in J \setminus \{j\}} w_j w_k \min\{\delta_{ij}, \delta_{ik}\} \: x_{ik} \qquad \forall j \in J, \forall i \in M.
\end{align*}
\end{lemma}
\begin{proof}
Fix a job $j$ assigned to machine $i^* \in M$ in the local optimum $x$ and let us assume that this job switches to machine $i \in M$. The difference of weighted completion times for job $j$ is
\[w_j \Big(p_{ij} + \sum_{k \prec_i j} p_{ik} x_{ik}\Big) - w_j \Big(p_{i^*j} + \sum_{k \prec_i^* j} p_{i^*k} x_{i^*k}\Big).\]
Moreover, the only other jobs for which the completion time is modified are the jobs assigned to $i^*$ and $i$ coming after $j$ in the ordering of the respective machine. Due to the switch of $j$, these jobs assigned to $i^*$ have their completion time decreased, whereas the ones assigned to $i$ have their completion time increased. The total difference in cost for these jobs is then
\[\sum_{k \succ_{i} j} w_k \: p_{ij} \: x_{ik} - \sum_{k \succ_{i^*} j} w_k \: p_{i^*j} \: x_{i^*k}.\]
Since $x$ is a local optimum for the global objective function, the total difference in cost (i.e. the sum of the two expressions above) should be non-negative. After rearranging the terms, this is equivalent to  
\[w_j \Big(p_{i^*j} + \sum_{k \prec_{i^*} j} p_{i^*k} x_{i^*k}\Big) + \sum_{k \succ_{i^*} j} w_k \: p_{i^*j} \: x_{i^*k} \leq w_j \Big(p_{ij} + \sum_{k \prec_i j} p_{ik} x_{ik}\Big) + \sum_{k \succ_{i} j} w_k \: p_{ij} \: x_{ik}.\]
Observe that this is exactly the statement of the lemma, finishing the proof.
\end{proof}
We now show that we can recover the tight approximation ratio of $(3 + \sqrt{5})/2$ given in \cite{correa2022performance} using our dual fitting approach. Observe the analogy with the proof strategy for the price of anarchy in the previous section. The main difference is that the Nash conditions are replaced by the local optimality conditions of Lemma \ref{lemma_local_opt}, and the $y$ variables are fitted differently.
\begin{theorem}
For any JumptOpt local optimum $x$ of the scheduling problem $R || \sum w_j C_j$, there exists a feasible (SDP-SR) solution with value at least $2/(3 + \sqrt{5}) \: C(x)$.
\end{theorem}
\begin{proof}
We assume that the SDP vectors belong to the space $\mathcal{F}(M)$, which is without loss generality by Lemma \ref{lemma_inner_product_space}. Let us fix $\alpha, \beta$ as in Lemma \ref{lemma_alpha_beta_cong}, i.e. $\alpha^2 := 2/\sqrt{5}$ and $\beta := 1/\alpha - \alpha/2$. We now state the dual fitting:
\begin{itemize}
\item $v_0(i,t) := \beta \sum_{k \in J} w_k \: x_{ik} \: \mathds{1}_{\{t \leq \delta_{ik}\}}$
\item $v_{ij}(i',t) := \alpha \; w_j \: \mathds{1}_{\{ t \leq \delta_{ij}\}} \: \mathds{1}_{\{ i = i' \}}  \hspace{2.5cm} \forall j \in J, \forall i \in \mathcal{S}_j$
\item $y_j := \alpha \beta \: \Big(w_j \: C_j(x) + D_j(x)\Big) \hspace{2.8cm} \forall j \in J.$
\end{itemize}
The desired inner products and norms can be computed to be the following, using essentially the same computations as in the proof of Theorem \ref{thm_Smith_Rule}:
\begin{align*}
\frac{1}{\beta^2}\:\Vert v_0 \Vert ^ 2 &= 2 C(x) - \eta(x) \qquad 
\frac{1}{\alpha \beta}\langle v_0, v_{ij} \rangle = \sum_{k \in J} w_j \: w_k \:  \min\{\delta_{ij}, \delta_{ik}\} \: x_{ik}\\
\frac{1}{\alpha^2}\:\Vert v_{ij} \Vert ^ 2 &= w_j \: p_{i j}  \hspace{1.6cm}
\frac{1}{\alpha^2}\: \langle v_{ij}, v_{i'k} \rangle = w_{j} \: w_{k} \: \min{\{\delta_{ij}, \delta_{ik}\}} \: \mathds{1}_{\{i = i'\}.}
\end{align*}
The second set of SDP constraints is satisfied due to the last computation above and the fact that $\alpha^2 \leq 1$. The first set of constraints under this fitting gives:
\begin{align*}
y_j &\leq w_j p_{ij} - \frac{1}{2}\Vert v_{ij} \Vert ^ 2 + \langle v_0, v_{ij} \rangle \\
&\iff \alpha \beta \: \Big( w_j \: C_j(x) + D_j(x) \Big) \leq \left(1 - \frac{\alpha^2}{2}\right) w_j \: p_{ij} + \alpha \beta \sum_{k \in J} w_j \: w_k \: x_{ik} \:  \min\{\delta_{ij}, \delta_{ik}\}.
\end{align*}
These are satisfied by Lemma \ref{lemma_alpha_beta_cong}, which states that $1 - \alpha^2/2 = \alpha \beta$, as well as the local optimality conditions of Lemma \ref{lemma_local_opt}. The objective function can now be lower bounded as:
\begin{align}
\sum_{j \in J} y_j  - \frac{1}{2} \Vert v_0 \Vert ^ 2 &= \alpha \beta \Big(2 C(x) - \eta(x)\Big) - \frac{\beta^2}{2}  \Big(2 C(x) - \eta(x)\Big) = \frac{2}{3 + \sqrt{5}}\Big(2 C(x) - \eta(x)\Big) \nonumber \\
&\geq \frac{2}{3 + \sqrt{5}} \: C(x) \label{eq_obj_sr}
\end{align}
where the first equality follows from \eqref{cost_1} and \eqref{cost_2}, the second equality follows from the second property of Lemma \ref{lemma_alpha_beta_cong} and the inequality follows from $\eta(x) \leq C(x)$.
\end{proof}
We now show as in \cite{correa2022performance} that one can get an improved bound for the restricted identical machines setting, denoted by $P | \mathcal{M}_j | \sum w_j C_j$. The improvement comes from the fact that for a \emph{JumpOpt} solution $x$ and an optimal solution $x^*$, we have $\eta(x) = \eta(x^*) = \sum_{j \in J} w_j p_j$ in this setting. This means that, instead of bounding $\eta(x) \leq C(x)$ in the last step of \eqref{eq_obj_sr}, we can now use the stronger upper bound $\eta(x) \leq C(x^*)$.

\begin{theorem}
For any JumptOpt local optimum $x$ of the scheduling problem $P | \mathcal{M}_j| \sum w_j C_j$, there exists a feasible (SDP-SR) solution with value at least $2/(3 + \sqrt{5}) \: (2 C(x) - C(x^*))$. By weak duality, this implies that the approximation ratio of $x$ is at most $(5 + \sqrt{5})/4 \approx 1.809$.
\end{theorem}
\begin{proof}
By upper bounding $\eta(x) \leq C(x^*)$ in the last step of \eqref{eq_obj_sr}, we get the first statement of the theorem. By weak duality, and since the dual solution constructed is feasible, we get that
\[\frac{2}{3 + \sqrt{5}} \Big(2 C(x) - C(x^*)\Big) \leq C(x^*) \iff \frac{C(x)}{C(x^*)} \leq \frac{5 + \sqrt{5}}{4}. \qedhere\]
\end{proof}

\subsection{An improved local search algorithm}
In this subsection, we show how our approach allows to analyze an improved local search algorithm for $R || \sum w_j C_j$ by \cite{caragiannis2017coordination} achieving an approximation ratio of $(5 + \sqrt{5})/4 + \varepsilon \approx 1.809 + \varepsilon$ for every $\varepsilon > 0$. To the best of our knowledge, this is the best currently known combinatorial approximation algorithm for this problem. We ignore here the issue of the running time and simply analyze the quality of a local optimum, referring the reader to \cite{caragiannis2017coordination} for further details. 
Let us fix the constant $\gamma := (9 + \sqrt{5})/19 \approx 0.591$. For each job $j \in J$ and an assignment $x$, we keep a potential function
\[f_j(x) = \sum_{i \in M} x_{ij} \left( w_j \: p_{ij} + \gamma \sum_{k \neq j} w_j w_k \min \{\delta_{ij}, \delta_{ik}\} \: x_{ik}\right) \quad \forall j \in J.\]
If a job $j \in J$ can pick a different machine than the one it is currently on and decrease its potential function $f_j(x)$, then this constitutes an improving move for the local search algorithm. If several improving moves exist, the algorithm picks the one giving the largest decrease in $f_j(x)$. For a local optimum $x$, we get the following constraints:
\begin{equation}
\label{eq_local_opt_glocal}
f_j(x) \leq w_j \: p_{ij} + \gamma \sum_{k \neq j} w_j w_k \min \{\delta_{ij}, \delta_{ik} \} \: x_{ik} \qquad \forall j \in J, \forall i \in M.
\end{equation}
As usual with this approach, we first need a small lemma about important constants.
\begin{lemma}
\label{lemma_alpha_beta_glocal}
Let $\alpha, \beta, \gamma \geq 0$ be defined as $\alpha^2 = (\sqrt{5}+1)/5, \: \beta^2 = (\sqrt{5}-1)/5$ and $\gamma = (9 + \sqrt{5})/19$. The following properties hold:
\begin{itemize}
\item $\alpha \beta / \gamma = 1 - \alpha^2/2$
\item $\alpha \beta (2 \gamma -1)/\gamma = \beta^2/2$
\item $2 \alpha \beta - \beta^2 = 4/(5 + \sqrt{5})$
\end{itemize}
\end{lemma}
\begin{proof}
The proof consists of simple computations and is omitted. These equations can also be checked on a computer.
\end{proof}
\begin{theorem}
For any local optimum $x$ of the above local search algorithm for $R || \sum w_j C_j$, there exists a feasible (SDP-SR) solution with value at least $4/(5 + \sqrt{5}) \: C(x)$.
\end{theorem}
\begin{proof}
We assume that the SDP vectors belong to the space $\mathcal{F}(M)$, which is without loss of generality by Lemma \ref{lemma_inner_product_space}. Let us fix $\alpha, \beta, \gamma$ as in Lemma \ref{lemma_alpha_beta_glocal}. We now state the dual fitting:
\begin{itemize}
\item $v_0(i,t) := \beta \sum_{k \in J} w_k \: x_{ik} \: \mathds{1}_{\{t \leq \delta_{ik}\}}$
\item $v_{ij}(i',t) := \alpha \; w_j \: \mathds{1}_{\{ t \leq \delta_{ij}\}}\mathds{1}_{\{ i = i' \}}  \hspace{2.5cm} \forall j \in J, \forall i \in \mathcal{S}_j$
\item $y_j := \frac{\alpha \beta}{\gamma} \: f_j(x) \hspace{5.1cm} \forall j \in J$.
\end{itemize}
The desired inner products and norms can be computed to be the following, using essentially the same computations as in the proof of Theorem \ref{thm_Smith_Rule}:
\begin{align*}
\frac{1}{\beta^2}\:\Vert v_0 \Vert ^ 2 &= 2 C(x) - \eta(x) \qquad 
\frac{1}{\alpha \beta}\langle v_0, v_{ij} \rangle = \sum_{k \in J} w_j \: w_k \:  \min\{\delta_{ij}, \delta_{ik}\} \: x_{ik}\\
\frac{1}{\alpha^2}\:\Vert v_{ij} \Vert ^ 2 &= w_j \: p_{i j}  \hspace{1.6cm}
\frac{1}{\alpha^2}\: \langle v_{ij}, v_{i'k} \rangle = w_{j} \: w_{k} \: \min{\{\delta_{ij}, \delta_{ik}\}} \: \mathds{1}_{\{i = i'\}.}
\end{align*}
The second set of SDP constraints is satisfied due to the last computation above and the fact that $\alpha^2 \leq 1$. The first set of constraints under this fitting gives:
\begin{align*}
y_j &\leq w_j p_{ij} - \frac{1}{2}\Vert v_{ij} \Vert ^ 2 + \langle v_0, v_{ij} \rangle \\
&\iff \frac{\alpha \beta}{\gamma} \: f_j(x) \leq \left(1 - \frac{\alpha^2}{2}\right) w_j \: p_{ij} + \frac{\alpha \beta}{\gamma} \gamma \: \sum_{k \in J} w_j \: w_k \:  \min\{\delta_{ij}, \delta_{ik}\} \: x_{ik}.
\end{align*}
These are satisfied by Lemma \ref{lemma_alpha_beta_glocal}, as well as the local optimality conditions \eqref{eq_local_opt_glocal}. To argue about the objective, it can be checked (through a simple computation that we omit) that:
\[\sum_{j \in J} f_j(x) = 2 \gamma \: C(x) - (2 \gamma -1) \: \eta(x).\]
The objective function then becomes:
\begin{align*}
\sum_{j \in J} y_j  - \frac{1}{2} \Vert v_0 \Vert ^ 2 &= \frac{\alpha \beta}{\gamma} \Big(2 \gamma \: C(x) - (2 \gamma -1) \: \eta(x)\Big) - \frac{\beta^2}{2}  \Big(2 C(x) - \eta(x)\Big)  \\ 
&= \Big(2 \alpha \beta - \beta^2 \Big) C(x) - \left(\frac{\alpha \beta \: (2 \gamma - 1)}{\gamma} - \frac{\beta^2}{2} \right) \eta(x) \\
& = \Big(2 \alpha \beta - \beta^2 \Big) C(x) = \frac{4}{5 + \sqrt{5}} \: C(x)
\end{align*}
where the two last equalities follow from Lemma \ref{lemma_alpha_beta_glocal}.
\end{proof}

We now provide an almost matching lower bound instance, inspired by constructions in \cite{caragiannis2006tight, correa2022performance}. We believe that the upper bound of $(5 + \sqrt{5})/4 \approx 1.809$ is tight.
\begin{theorem}
There exists an instance of $R || \sum w_j C_j$ with a local optimum to the above local search algorithm with approximation ratio at least $1.791$.
\end{theorem}
\begin{proof}
Let $\lambda \approx 1.33849$ be the positive solution to the equation $\lambda^2 = 1 + \gamma \: \lambda$. We consider an instance with jobs $J = [n]$ and machines $M = [n+1]$. The weights of the jobs are defined as $w_1 = \lambda$ and $w_j = 1/\lambda^{j-1}$ for every $j \geq 2$. The feasible machines are $\mathcal{S}_j = \{j ,j+1\}$ for every $j \in J$ with processing times $p_{1,1}, p_{2,1} = \lambda$ for the first job and $p_{j,j} = \lambda^{j-1}, p_{j+1, j} = \lambda^{j+1}$ for every $j \geq 2$.

The feasible solution where each job $j$ gets assigned to machine $j$ has cost $\sum_{j \in J} w_j \: p_{j,j} = \lambda^2 + (n-1)$, showing that the optimum solution $x^*$ satisfies $C(x^*) \leq n - 1 + \lambda^2$.

We now claim that the solution $x$ where each job $j$ gets assigned to machine $j+1$ is a local optimum. To see this, observe that the first job clearly cannot decrease his potential function $f_1(x)$ since $p_{1,1} = p_{2,1}$ and no other job is assigned to machine 1 or 2. For $j \geq 2$, we have $f_j(x) = w_j \: p_{j+1, j} = \lambda^2$. If job $j$ were to be reassigned to machine $j$, then \[f_j(x_{-j}, j) = w_j \: p_{j,j} + \gamma \: w_{j-1} w_{j} \min\{\delta_{j, j-1}, \delta_{j,j}\} = 1 + \gamma \: \lambda,\]
which shows that $x$ is a local optimum, by definition of $\lambda$. The cost of this solution is then $\sum_{j \in N} \: w_j \: p_{j+1 ,j} = n \lambda^2$. The approximation ratio of this solution now satisfies \[\frac{C(x)}{C(x^*)} \geq \frac{n \lambda^2}{n-1 + \lambda^2} \xrightarrow{n \to \infty} \lambda^2 \approx 1.79154.\]
Picking $n$ large enough thus finishes the proof.
\end{proof}

\section{Analyzing a greedy online algorithm}
\label{sec_greedy_online}
Consider the following \emph{online} problem for our congestion model: each $j \in N$ arrives online one by one in an adversarial order and reveals its strategy set $\mathcal{S}_j \subseteq 2^E$, its processing times $p_{ej} \geq 0$ for every $e \in E$ and its weight $w_j \geq 0$, at which point an online algorithm needs to irrevocably pick $i \in \mathcal{S}_j$, i.e. set $x_{ij} = 1$ for some $i \in \mathcal{S}_j$. The goal of the algorithm is to minimize the total incurred cost $C(x)$ defined in \eqref{eq_soccost_selfrout}. We denote the total incurred cost after the arrival of $j \in N$ by $C^{(j)}(x)$. The quality of an online algorithm is measured by the \emph{competitive ratio}: the worst-case ratio, over all instances, between the cost of the algorithm and that of the optimal offline solution. Since jobs/players arrive online, we denote them as $N = \{1, \dots, n\}$, where the notation $j < k$ indicates that $j$ arrives before $k$.

\begin{algorithm}
    \caption{Greedy algorithm}
    \label{greedy_algo_LB}
    \begin{algorithmic}
    \State \textbf{when $j \in N$ arrives:}
    \State \qquad Set $x_{ij} = 1$ for $i \in \mathcal{S}_j$ giving the minimal increase in the objective function $C(x)$
    \State \Return $x$
 \end{algorithmic}
 \end{algorithm}

We consider and analyze the above algorithm named \Call{Greedy}{}. Whenever a job $j \in \J$ arrives, \Call{Greedy}{} picks $i \in \mathcal{S}_j$ (i.e. sets $x_{ij} = 1$) which gives the least increase in the global objective function. It is known that this algorithm is $4$-competitive for the scheduling problem $R || \sum w_j C_j$ \cite{gupta2020greed}, we show here how to recover this bound in our congestion model.\footnote{In the conference version of this paper, we also present a different $4$-competitive randomized algorithm.} The key property of the greedy algorithm is the following lemma.
\begin{lemma}
\label{lemma_greedy_SR}
For any adversarial instance of the above online problem, and any solution $(x_{ij})_{j \in \J, i \in \mathcal{S}_j}$ obtained by \Call{Greedy}{}, the following inequalities are satisfied for all $j \in \J$:
\[C^{(j)}(x) - C^{(j-1)}(x) \leq \sum_{e \in i}\left(w_j \: p_{ej} + \sum_{k < j} w_j \: w_k \: \min\{\delta_{ej}, \delta_{ek}\} \: z_{ek} \right) \qquad \forall i \in \mathcal{S}_j.\]
\end{lemma}
\begin{proof}
Consider the online arrival of $j \in \J$. At that moment in time, the total cost summed over all resources is $C^{(j-1)}(x)$.
Let us analyze the increase in cost if $j$ were to pick any $i \in \mathcal{S}_j$. For every resource $e \in i$, the weighted completion time of $j$ gives a contribution of 
\[w_j \Big(p_{ej} + \sum_{k < j, k \prec_e j} p_{ek} \: z_{ek} \Big).\]
Moreover, the only jobs for which the completion time is modified on that resource are the already arrived jobs $k < j$ assigned to that resource which have a higher Smith ratio, since their completion time is pushed further by $p_{ej}$ due to the entrance of $j$. Hence, the increase in objective due to those jobs is:
\[\sum_{k < j, k \succ_e j} w_k \: p_{ej} \: z_{ek}.\]
Now, observe that by definition of the Smith ratio $\delta_{ek} = p_{ek}/w_k$, the total increase in objective on every $e \in i$ (i.e. the sum of the two above quantities) can be written as:
\[ w_j \: p_{ej} + \sum_{k < j} w_j \: w_k \: \min\{\delta_{ej}, \delta_{ek}\} \: z_{ek} \qquad \forall e \in i.\]
The total increase in cost then sums this quantity over every resource $e \in i$. By definition, \Call{Greedy}{} will pick $i \in \mathcal{S}_j$ which gives the smallest increase in the objective function, leading to the statement of the lemma.
\end{proof}
We are now ready to analyze the competitive ratio of \Call{Greedy}{}. We will do so by doing a dual fitting argument on the \emph{(SDP-SR)} semidefinite program, that we rewrite below for convenience.

\begin{align}
\label{sdp_sr_copied}
    \max \sum_{j \in \J} y_j  - \frac{1}{2}  &\Vert \vb \Vert ^ 2 \\
    y_j &\leq \sum_{e \in i} w_j \: p_{ej} - \frac{1}{2}\Vert v_{ij} \Vert ^ 2 + \langle \vb, v_{ij} \rangle \qquad \qquad \forall j \in \J, \forall i \in \mathcal{S}_j \nonumber\\
    \langle v_{ij}, v_{i'k} \rangle &\leq \sum_{e \in i \cap i'} w_{j} \: w_{k} \: \min{\{\delta_{ej}, \delta_{ek}\}}\hspace{2.2cm} \forall (i,j) \neq (i',k) \text{ with } j,k \in \J. \nonumber
\end{align}

The proof of the theorem below is very similar to the one for the price of anarchy of Smith's Rule in Theorem \ref{thm_Smith_Rule}. The main difference is that the Nash equilibrium inequalities are replaced by the inequalities of Lemma \ref{lemma_greedy_SR}, and the $y$ variables are fitted differently. Note that the vector fitting stays the same.

\begin{theorem}
\label{thm_greedy_SR}
For any instance of the above online problem, and any solution $(x_{ij})_{j \in \J, i \in \mathcal{S}_j}$ obtained by \Call{Greedy}{}, there exists a feasible (SDP-SR) solution with objective value at least $C(x)/4$. By weak duality, this implies that the competitive ratio of \Call{Greedy}{} is at most $4$.
\end{theorem}
\begin{proof}
We assume that the SDP vectors live in the inner product space $\mathcal{F}(\M)$, which is without loss of generality by Lemma \ref{lemma_inner_product_space}. 
Let us fix $\beta = 1/2$, we now state the dual fitting for \emph{(SDP-SR)}:
\begin{itemize}
\item $\vb(e,t) := \beta \sum_{k \in \J} w_k \: z_{ek} \: \mathds{1}_{\{t \leq \delta_{ek}\}}$
\item $v_{ij}(e,t) := w_j \: \mathds{1}_{\{ e \in i\}} \: \mathds{1}_{\{ t \leq \delta_{ej}\}}  \hspace{3.2cm} \forall j \in \J, \forall i \in \mathcal{S}_j$
\item $y_j := \beta \: \Big(C^{(j)}(x) - C^{(j-1)}(x)\Big) \hspace{3cm} \forall j \in \J.$
\end{itemize}
Using the same computations as in the proof of Theorem \ref{thm_Smith_Rule}, we get:
\begin{align*}
\frac{1}{\beta^2}\:\Vert v_0 \Vert ^ 2 &\leq \:  2 C(x) \hspace{1.3cm} 
\frac{1}{ \beta}\langle v_0, v_{ij} \rangle = \sum_{e \in i} \sum_{k \in N} w_j \: w_k \: z_{ek} \:  \min\{\delta_{ej}, \delta_{ek}\}\\
\Vert v_{ij} \Vert ^ 2 &= \sum_{e \in i} w_j \: p_{e j}  \hspace{0.9cm}
\langle v_{ij}, v_{i'k} \rangle = \sum_{e \in i \cap i'} w_{j} \: w_{k} \: \min{\{\delta_{ej}, \delta_{ek}\}}.
\end{align*}
Let us now check that this is a feasible solution to \emph{(SDP-SR)}. The second set of constraints is tightly satisfied due to the last computation above. The first set of constraints under the above fitting becomes:
\begin{align*}
y_j \leq \sum_{e \in i} w_j \: p_{ej} - \frac{1}{2}\Vert v_{ij} \Vert ^ 2 & + \langle \vb, v_{ij} \rangle \\
\iff \beta \: \Big(C^{(j)}(x) &- C^{(j-1)}(x)\Big) \leq \frac{1}{2} \sum_{e \in i} w_j \: p_{ej} + \beta \sum_{e \in i} \sum_{k \in \J} w_j \: w_k \: z_{ek} \:  \min\{\delta_{ej}, \delta_{ek}\}.
\end{align*}
By the choice $\beta = 1/2$, these inequalities are satisfied by Lemma \ref{lemma_greedy_SR}, implying that the fitted solution is feasible. To argue about the objective function, note that the sum of the $y$ variables becomes:
\[\sum_{j \in \J}y_j = \beta \sum_{j \in \J} \Big(C^{(j)}(x) - C^{(j-1)}(x)\Big)= \beta \: C(x)\]
where the second equality uses the fact that the sum is telescoping.
The objective function of this SDP can now be lower bounded as:
\begin{align*}
\sum_{j\in N} y_j  - \frac{1}{2} &\Vert \vb \Vert ^ 2 \geq \beta C(x) - \beta^2 C(x) = \frac{C(x)}{4}. \qedhere
\end{align*}
\end{proof}

\section{Concluding remarks}
In this paper, we built on the work of \cite{kulkarni2014robust} which showed a way to use convex programming duality to prove price of anarchy bounds for different games. We showed that a unique convex program turns out to be surprisingly powerful and allows to get tight upper bounds for a large class of congestion and scheduling games. Moreover, it can also be used to bound the approximation ratio of local search algorithms and the competitive ratio of online algorithms for such problems. The dual program has a simple structure with the first set of constraints being similar to equilibrium inequalities, guiding the dual fitting approach. This program also has a natural connection to the inner product space structure developed in \cite{cole2011inner}. It would be interesting if this approach can be extended to new games where price of anarchy bounds are not yet settled. Moreover, all the problems we considered had a quadratic (possibly non-convex) objective function, which made the first round of the Lasserre/Sum of Squares SDP hierarchy powerful enough to write a tractable convex relaxation. It would be interesting if a similar technique can work for problems with a higher degree polynomial objective (an example of which are congestion games with polynomial latency functions) by considering later rounds of the hierarchy. To the best of our knowledge, such dual fitting arguments on semidefinite programs have not been explored much: we hope and believe that there may be additional applications to such an approach.
\subsubsection*{Acknowledgements}
The author would like to thank Zhuan Khye Koh for useful discussions. The author would also like to thank Guido Schäfer for useful discussions, as well as nice feedback on an earlier version of this document. The author also thanks an anonymous reviewer who pointed out that the online greedy algorithm has been analyzed in \cite{gupta2020greed}.
\bibliographystyle{alpha}
\bibliography{references}

\appendix
\section{Recovering the Kawaguchi-Kyan bound for $P || \sum w_j C_j$}
\label{sec_kawa_kyan}
In this section, we show that we can recover the optimal bound of $(1 + \sqrt{2})/2$ for the pure price of anarchy of the scheduling game on parallel machines $P || \sum w_j C_j$, where each machine uses increasing Smith ratios to schedule the jobs. To do so, we make use of a sequence of reductions to worst-case instances provided in \cite{schwiegelshohn2011alternative}. The first assumption that we can make is that $w_j = p_j$ for every job $j$. The \emph{(SDP-SR)} dual semidefinite program shown in \eqref{sdp_sr} and used in Section \ref{sec_local_opt} for $R || \sum w_j C_j$ in this special case becomes the following. We denote the set of jobs by $J$ and the set of machines by $M$.

\begin{align*}
    \max \sum_{j \in J} y_j  - \frac{1}{2}  &\Vert v_0 \Vert ^ 2 \\
    y_j &\leq p_j^2 - \frac{1}{2}\Vert v_{ij} \Vert ^ 2 + \langle v_0, v_{ij} \rangle \qquad \qquad \forall j \in J, \forall i \in M\\
    \langle v_{ij}, v_{i'k} \rangle &\leq  p_j \: p_k \:  \mathds{1}_{\{i = i'\}} \hspace{3cm} \forall (i, j) \neq (i',k) \text{ with } j,k \in J.
\end{align*}

Moreover, the reduction in \cite{schwiegelshohn2011alternative} states that we may assume the instance only has two different processing times $\varepsilon, p > 0$, where $\varepsilon$ is an arbitrarily small constant. Jobs with processing time $\varepsilon$ are called small jobs, and the total workload of these jobs is $|M|$, i.e. the total number of small jobs is $|M|/\varepsilon$. Jobs with processing times $p$ are called large jobs and the total number of large jobs is $k < |M|$, i.e. strictly less than the number of machines. In addition, in a pure Nash equilibrium $x$:
\begin{itemize}
\item All small jobs are started and completed in the interval $[0,1]$.
\item All large jobs are started at $1$.
\end{itemize}
In this reduced instance, it is also possible to get an exact expression for the optimum solution. In particular, define $\alpha := m/(m-k)$ and $\beta := (m + pk)/m$, an optimal solution $x^*$ then has cost:
\[C(x^*) = \begin{cases} kp^2 + \frac{m-k}{2} \alpha^2 \hspace{3cm} \text{if } p \geq \alpha \\
\frac{1}{2}\Big(k p^2 + m \beta^2\Big) \hspace{2.77cm} \text{if } p \leq \alpha. \end{cases}\]
It can then be shown that in both cases $C(x)/C(x^*) \leq (1 + \sqrt{2})/2$ through a simple calculus analysis. The reader is referred to \cite{schwiegelshohn2011alternative} for details. 

We show here that we can construct a feasible dual solution to the SDP matching the objective value of $C(x^*)$, showing that the SDP does not have an integrality gap on such a reduced instance and thus implying that the price of anarchy is at most $(1 + \sqrt{2})/2$ by a dual fitting proof.

\begin{theorem}
For any instance of the above game on the reduced instance, there exists a feasible (SDP-SR) solution with objective value $C(x^*)$, implying that the pure price of anarchy is at most $(1 + \sqrt{2})/2$.
\end{theorem}
\begin{proof}
The vectors in our dual fitting will live in the space $\mathbb{R}^M$. Let us denote the total number of machines by $m = |M|$, and let us set $\alpha := m/(m-k)$. We denote by $\mathds{1}$ the all ones vector and by $e_i$ the $i^{th}$ standard basis vector. 

We now state the dual fitting for the case where $p \geq \alpha$:
\begin{itemize}
\item $v_0 = \alpha \: \mathds{1}$
\item If $j$ is a large job, then set $v_{ij} = \alpha \: e_i$ and $y_j = p^2 + \alpha^2/2$
\item If $j$ is a small job, then set $v_{ij} = \varepsilon \: e_i$ and $y_j = \varepsilon \alpha$
\end{itemize}
Let us check that this solution is indeed feasible. Clearly, if $i \neq i'$, then $\langle v_{ij}, v_{i'k} \rangle = 0$ by orthogonality of $e_i$ and $e_{i'}$. For two jobs $j \neq k$, we have that $\langle v_{ij}, v_{ik} \rangle \leq \Vert v_{ij} \Vert \: \Vert v_{ik} \Vert \leq p_j p_k$ where we use Cauchy-Schwarz for the first inequality and the fact that $\alpha \leq p$ if some job is large for the second inequality. This shows that the second set of SDP constraints is satisfied.

Moreover, the first set of constraints is satisfied as well, as the SDP inequalities yield $y_j \leq p^2 + \alpha^2/2$ for large jobs and $y_j \leq \varepsilon^2/2 + \varepsilon \alpha$ for small jobs, which is satified by our choice of $y_j$. The objective value of this dual solution is then:
\begin{align*}
\sum_{j \in J} y_j - \frac{1}{2} \Vert v_0 \Vert ^ 2 &= k \left(p^2 + \frac{\alpha^2}{2}\right) + \frac{m}{\varepsilon}\varepsilon \alpha - \frac{1}{2}m \alpha^2 = kp^2 + \frac{m-k}{2} \alpha^2
\end{align*}
where the first equality follows since the number of small jobs is $m/ \varepsilon$ and the last equality follows by observing that $m \alpha = (m-k)\alpha^2$ by definition of $\alpha$.

For the case where $p \leq \alpha$, we define $\beta:= (m + pk)/m$. We now state the dual fitting:
\begin{itemize}
\item $v_0 = \beta \: \mathds{1}$
\item If $j$ is a large job, then set $v_{ij} = p \: e_i$ and $y_j = p^2/2 + \beta \: p$
\item If $j$ is a small job, then set $v_{ij} = \varepsilon \: e_i$ and $y_j = \varepsilon \beta$
\end{itemize}
Similarly to before, the second set of constraints is satisfied by orthogonality of the standard basis vectors and the fact that $\Vert v_{ij} \Vert = p_j$ for all jobs (either small or large). The first set of constraints yields $y_j \leq p^2/2 + \beta p$ for large jobs and $y_j \leq \varepsilon^2/2 + \varepsilon \beta$ for small jobs, which is clearly satisfied by our choice of $y_j$. The objective value of this dual solution is then:
\[\sum_{j \in J} y_j - \frac{1}{2} \Vert v_0 \Vert ^ 2 = k \left(\frac{p^2}{2} + \beta p\right) + \frac{m}{\varepsilon} \varepsilon \beta - \frac{m \beta^2}{2} = \frac{kp^2}{2} + \beta(kp + m) - \frac{m \beta^2}{2} = \frac{1}{2}(k p^2 + m \beta^2)\]
where the last equality follows by observing that $m \beta^2 = (m+pk)\beta$ due to the definition of $\beta$.
\end{proof}

\section{Preliminaries on SDPs}
\label{sec_sdp_basics}
A symmetric matrix $X \in \mathbb{R}^{n \times n}$ is \emph{positive semidefinite}, denoted as $X \succeq 0$, if the following equivalent conditions hold:
\begin{enumerate}
\item $x^T X x \geq 0$ for all $x \in \mathbb{R}^n$
\item All the eigenvalues of $X$ are non-negative
\item There exists vectors $v_1, \dots, v_n \in \mathbb{R}^d$ for some $d>0$ such that $X_{ij} = \langle v_i, v_j \rangle$ for all $i,j \in [n]$.
\end{enumerate}
For $A, B \in \mathbb{R}^{n \times n}$, the \emph{trace inner product} is defined as:
\[\langle A, B \rangle := \text{Tr}(A^TB) = \sum_{i,j = 1}^n A_{ij} \: B_{ij}.\]
Given symmetric matrices $A_1, \dots, A_m \in \mathbb{R}^{n \times n}$ and $b \in \mathbb{R}^m$, a \emph{semidefinite program (SDP)} in standard form is the following optimization problem:
\[p^* = \sup_X \left\{\langle C,X \rangle : \langle A_k, X \rangle = b_k \; (k \in [m]) \;, X \succeq 0\right\}.\]
Each SDP of that form admits a dual SDP program:
\[d^* = \inf_y \left\{b^T y : Y = \sum_{k = 1}^m y_k A_k - C, Y \succeq 0\right\}.\]
Weak duality holds, meaning that $p^* \leq d^*$. By Property 3 described above, in order to come-up with a feasible dual solution, it is enough to construct $y \in \mathbb{R}^m$, as well as vectors $v_1, \dots, v_n \in \mathbb{R}^d$ in some dimension $d > 0$ such that $Y_{ij} = \left(\sum_{k = 1}^m y_k A_k - C \right)_{ij} = \langle v_i, v_j \rangle$ for every $i,j \in [n]$.
\section{Computation of the dual SDPs}
\subsection{Taking the dual}
\label{sec_take_dual}
Recall that our primal semidefinite programming relaxation is the following.
\begin{align*}
    \min \langle C, X \rangle \qquad  \qquad &\\
    \sum_{i \in \mathcal{S}_j} X_{\{ij,\: ij\}} &= 1 \hspace{4cm} \forall j \in N \\
    X_{\{0,0\}} & = 1 \\
    X_{\{0, \: ij\}} &= X_{\{ij, \: ij\}} \hspace{3cm} \forall j \in N, i \in \mathcal{S}_j \\
    X_{\{ij, \: i'k\}} &\geq 0 \hspace{4cm} \forall {(i,j), (i',k)} \text{ with } j,k > 0. \\
    X & \succeq 0
\end{align*}
It can be easily checked that the following form of semidefinite programs is a primal-dual pair. The dual variables $(\lambda_i)_i$ and $(\mu_j)_j$ respectively correspond to the equality and inequality constraints, whereas the matrix variable $Y$ corresponds to the semidefinite constraint.

\vspace{0.3cm}
\begin{minipage}[t]{0.4\textwidth}
\begin{align*}
    \min \langle C, X \rangle & \\
    \langle A_i, X \rangle & = b_i \quad \forall i \\
    \langle B_j, X \rangle & \geq 0 \quad \: \forall j\\
    X &\succeq 0 
\end{align*}
\end{minipage}
\begin{minipage}[t]{0.6\textwidth}
\begin{align*}
    \max \sum_{i} b_i &\lambda_i \\
    Y &= C - \sum_i \lambda_i A_i  - \sum_j \mu_j B_j \\
    Y &\succeq 0, \quad \mu \geq 0.
\end{align*}
\end{minipage}
\vspace{0.3cm}

\noindent Observe that our above primal SDP is in fact of that form. Let us denote by $(y_j)_{j \in N}, z$ and $(\sigma_{ij})_{j \in N, i \in \mathcal{S}_j}$ the dual variables respectively corresponding to the three sets of equality constraints. Let us denote by $\mu_{\{ij, i'k\}} \geq 0$ the dual variables corresponding to the inequality (or non-negativity) constraints. The dual objective then becomes $\sum_{j \in N} y_j + z$. 

All the games considered will satisfy the fact that the objective matrix is all zeros in the first row and column:  $C_{\{0,0\}} = 0$ and $C_{\{0,ij\}} = 0$ for every $j \in N$ and $i \in \mathcal{S}_j$. The dual matrix equality then becomes:
\begin{align*}
Y_{\{0,0\}} &= -z \\
Y_{\{0, \: ij\}} &= \frac{\sigma_{ij}}{2} \hspace{6cm} \forall j \in N, i \in \mathcal{S}_j\\
Y_{\{ij, \: ij\}} &= C_{\{ij, \; ij\}} - y_j - \sigma_{ij} - \mu_{\{\ij, \: \ij\}} \hspace{2.1cm} \forall j \in N, i \in \mathcal{S}_j\\
Y_{\{ij, \: i'k\}} &= C_{\{ij, \: i'k\}} - \mu_{\{ij,\:  i'k\}} \hspace{3.5cm} \forall (i,j) \neq (i',k) \text{ with } j,k >0.
\end{align*}
Note that we can now eliminate the dual variables $z$ and $\sigma$ by the first two equalities. Moreover, we can eliminate the $\mu \geq 0$ variables by replacing the last two equalities by inequalities. Let us now do the change of variable $Y' = 2Y$ and let the vectors of the Cholesky decomposition of $Y'$ be $v_0$ and $(v_{ij})_{j \in N, i \in \mathcal{S}_j}$, meaning that $Y'_{a, \: b} = \langle v_{a}, v_{b} \rangle$ holds for all the entries of $Y'$. The dual SDP in vector form can then be rewritten as:
\begin{align*}
    \max \sum_{j \in N} y_j  - &\frac{1}{2}\Vert v_0 \Vert ^ 2 \\
    y_j &\leq C_{\{ij, \; ij\}} - \frac{1}{2}\Vert v_{ij} \Vert ^ 2 - \: \langle v_0, v_{ij} \rangle \qquad \qquad \forall j \in N, i \in \mathcal{S}_j\\
    \langle v_{ij}, v_{i'k} \rangle &\leq 2 \: C_{\{ij, \: i'k\}} \hspace{4.3cm} \forall (i,j) \neq (i',k) \text{ with } j,k >0
\end{align*}
Clearly, we can also do a change of variable $v_0' := - v_0$ to get the program \eqref{dual_sdp}.

\subsection{Specializing it to the different games considered}
\label{sec_comp_dual}
Let us now describe how the objective matrix $C$ looks like for the different games that we need. Recall from Section \ref{sec_SDP} that we need to pick a symmetric matrix $C$ such that $C(x) = \langle C, X \rangle = \text{Tr}(C^TX)$ where $X = (1, x) (1, x)^T$ is a binary rank one matrix and $C(x)$ is the social cost. By definition of the trace inner product, this is equivalent to:
\[C(x) = C_{\{0,0\}} + 2\sum_{j \in N, i \in \mathcal{S}_j} C_{\{0, ij\}}\:x_{ij} + \sum_{\substack{j, k \in N\\ i \in \mathcal{S}_j, i' \in \mathcal{S}_k}} C_{\{ij, \: i'k\}}x_{ij} \: x_{i'k}.\]
Recall also that $x_{ij}^2 = x_{ij}$ since $x_{ij} \in \{0,1\}$. Hence, if the social cost does not have constant terms, we will always be able to pick $C$ such that $C_{\{0,0\}} = 0$ and $C_{\{0, ij\}} = 0$ for every $j \in N, i \in \mathcal{S}_j$, which we do for all the games below.

For the congestion game under the \emph{Smith Rule} policy,
% we have seen that
%\[C_j(x) = \sum_{i \in S_j} x_{ij} \sum_{e \in i}  \Big( p_{ej} +  \sum_{k \prec_e j} p_{ek} \: z_{ek} \Big),\]
%where $z_{ej}:= \sum_{i \in S_j : e \in i} x_{ij}$. Note that this is in fact a quadratic function of $x$, by plugging in the definition of $z$:
%\[C_j(x) = \sum_{i \in \mathcal{S}_j} \sum_{e \in i} p_{ej} \: x_{ij} +  \sum_{i \in \mathcal{S}_j} \sum_{e \in i} \sum_{k \prec_e j} \: \sum_{i' \in \mathcal{S}_k : e \in i'} p_{ek} \: x_{ij} \: x_{i'k}.\]
the social cost in \eqref{eq_soccost_selfrout} can be written as:
\[C(x) = \sum_{j \in N} \: w_j \: C_j(x) = \sum_{\substack{j \in N \\ i \in \mathcal{S}_j \\ e \in i}} w_j \: p_{ej} \: x_{ij}+ \frac{1}{2}\sum_{\substack{j \in N, k \neq j\\ i \in \mathcal{S}_j, i' \in \mathcal{S}_k \\ e \in i \cap i'}} w_j \: w_k \min\{\delta_{ej}, \delta_{ek}\} \: x_{ij} \: x_{i'k}.\]
Therefore, the objective matrix $C$ is the following: 
\[C_{\{ij, \: ij\}} = \sum_{e \in i} w_j \: p_{ej} \quad, \quad C_{\{ij, \: i'k\}} = \frac{1}{2} \sum_{e \in i \cap i'} w_j \: w_k \: \min\{\delta_{ej}, \delta_{ek}\}.\]
If one considers the scheduling problem $R || \sum w_j C_j$ under \emph{Smith's Rule}, which is a special case of the previous setting, then
\[C_{\{ij, \: ij\}} = w_j \: p_{ij} \quad, \quad C_{\{ij, \: i'k\}} = \frac{1}{2} \: w_j \: w_k \: \min\{\delta_{ij}, \delta_{ik}\} \: \mathds{1}_{\{i = i'\}}.\]
For the congestion game under the \emph{Rand} policy, the social cost in \eqref{eq_soccost_rand} gives
\[C_{\{ij, \: ij\}} = \sum_{e \in i} w_j \: p_{ej} \quad, \quad C_{\{ij, \: i'k\}} = \sum_{e \in i \cap i'} w_j \: w_k \: \frac{\delta_{ej} \delta_{ek}}{\delta_{ej} + \delta_{ek}}.\]
For the weighted affine congestion game, we have seen that 
\[C(x) := \sum_{j \in N} C_j(x) = \sum_{e \in E} a_e \: \ell_e(x)^2 + b_e \:  \ell_e(x)\]
where $\ell_e(x) = \sum_{j \in N} w_{ej} \sum_{i \in \mathcal{S}_j : \: e \in i} x_{ij}$. The objective matrix in that case is
\[C_{\{ij, \: ij\}} = \sum_{e \in i} w_{ej}(a_e \: w_{ej} + b_e) \quad, \quad C_{\{ij, \: i'k\}} = \sum_{e \in i \cap i'} a_e \: w_{ej} \: w_{ek}.\]

\section{Robust price of anarchy}
\label{section_robust}
In this section, we describe how our proofs can be adapted to give bounds on the \emph{coarse-correlated price of anarchy}, meaning that we can now generalize our results by considering \emph{coarse-correlated equilibria}, instead of \emph{mixed} (or \emph{pure}) Nash equilibria. 

Let $N$ be a game with a strategy set $\mathcal{S}_j$ and payoff function $C_j$ for every player $j \in N$. A distribution $\sigma$ over $\mathcal{S}_1 \times \dots \times \mathcal{S}_n$ is a coarse correlated equilibrium if 
\begin{equation}
\label{eq_coarse_corr}
\mathbb{E}_{X \sim \sigma} [C_j(X)] \leq \mathbb{E}_{X \sim \sigma} [C_j(X_{-j}, i)] \qquad \forall j \in N, i \in \mathcal{S}_j.
\end{equation}
Note that this generalizes a mixed Nash equilibrium. In that case, $\sigma$ is a product distribution, i.e. every player $j$ picks a random strategy independently from its own distribution, which we denoted by $(x_{ij})_{i \in \mathcal{S}_j}$ previously in the paper. We note that our formulas for $C_j(x)$ - see for instance \eqref{eq_SR_formula} - for non-binary $x$ (i.e. interpreting $x$ as a collection of probability distributions rather than an integer assignment) implicitly use this independence assumption, meaning that the current proofs do not directly go through for coarse-correlated equilibria. Let us first rewrite \emph{(SDP-C)} \eqref{dual_sdp} in a more convenient matrix form for this argument.

\vspace{0.1cm}
\begin{align}
    \max \sum_{j \in N} \varphi_j  - &\frac{1}{2} \: Y_{\{0,0\}} \nonumber\\
    \varphi_j &\leq C_{\{ij, \; ij\}} - \frac{1}{2} \: Y_{\{ij, ij\}} + \: Y_{\{0, ij\}} \qquad \qquad \forall j \in N, i \in \mathcal{S}_j \nonumber \\
    Y_{\{ij, i'k\}} &\leq 2 \: C_{\{ij, \: i'k\}} \hspace{4.3cm} \forall (i,j) \neq (i',k) \text{ with } j,k \in N. \nonumber \\
    Y &\succeq 0 \nonumber
\end{align}
\vspace{-0.2cm}

One way to generalize our results is to consider random dual \emph{(SDP-C)} solutions, i.e. doing a dual fitting on a realization $X \sim \sigma$, which induces binary random variables $\{X_{ij}\}_{j \in N, i \in \mathcal{S}_j}$ and $\{Z_{ej}\}_{j \in N, e \in E}$. For any price of anarchy dual fitting argument in this paper, first replace every occurence of respectively $x_{ij}$ and $z_{ej}$ by $X_{ij}$ and $Z_{ej}$, in which case $v_0$ and every $y_j$ become random variables (note that every $v_{ij}$ is always deterministic). To get a feasible dual solution, we now set $Y_{a,b} := \mathbb{E}_{X \sim \sigma}[\langle v_a, v_b \rangle$] for every indices $a,b$ as well as $\varphi_j : = \mathbb{E}_{X \sim \sigma}[ \: y_j \: ]$.

The second set of constraints of \emph{(SDP-C)} is always satisfied deterministically in our fittings, while the first set of constraints is satisfied by considering expectations, due to inequality \eqref{eq_coarse_corr}. Moreover, $Y$ is positive semidefinite since it is a convex combination of positive semidefinite matrices.

We thus get a feasible solution with objective value $V$ satisfying $V \geq \rho \: \mathbb{E}_{X \sim \sigma}[C(X)]$ for some desired bound $\rho \in [0,1]$. Since the dual solution is feasible, we have $V \leq C(x^*)$, where $x^*$ is the social optimum. Combining these two equations gives:
\[\mathbb{E}_{X \sim \sigma} [C(X)] \leq \frac{1}{\rho} C(x^*)\]
hence yielding a bound on the coarse-correlated price of anarchy.

\end{document}